\documentclass[10pt,aps,prd,twocolumn,nofootinbib,floatfix]{revtex4}
\usepackage{amsfonts,amsmath,amssymb,amsthm}
\usepackage{graphicx}
\usepackage{color}

\usepackage{enumerate}

\newtheorem{thm}{Theorem}
\newtheorem{lem}[thm]{Lemma}
\newtheorem*{remark}{Remark}

\begin{document}

\author{Patryk Mach}\email{patryk.mach@uj.edu.pl}
\affiliation{Institute of Theoretical Physics, Jagiellonian University, {\L}ojasiewicza 11, 30-348 Krak\'{o}w, Poland}
\author{Edward Malec}\email{edward.malec@uj.edu.pl}
\affiliation{Institute of Theoretical Physics, Jagiellonian University, {\L}ojasiewicza 11, 30-348 Krak\'{o}w, Poland}
 
\title{Steady critical accretion onto black holes: selfgravity and sonic point characteristics}

\begin{abstract}
The spherically symmetric steady accretion  of polytropic perfect fluids onto a black hole is the simplest flow model that can demonstrate effects of backreaction (selfgravity). It has been discovered 16 years ago that backreaction does not influence some (``intensive'') characteristics of sonic points, under suitable conditions. Herein we consider a wider class of equations of state, with polytropic indices in the range $(1,2]$, and establish detailed boundary conditions that allow one to prove this fact. We find also  numerical examples showing limits of our analytic criteria---if suitable analytic conditions are not satisfied, then selfgravity influences all characteristics of sonic points. That fact constrains the applicability of the recent proposal of Baumgarte and Shapiro to estimate the lifetime of black holes within compact stellar objects. 
\end{abstract}

\maketitle

\section{Introduction}

The spherical steady accretion of a massless gas onto a gravitational center has been investigated in the newtonian context by Bondi \cite{bondi} and in the Schwarzschild space-time by  Michel \cite{michel}, Shapiro and Teukolsky \cite{shapiro_teukolsky}. The general-relativistic description including backreaction has been formulated in \cite{malec} and investigated in \cite{PRD2006, KM2006, Mach2007,Mach2008, Mach2009, KMRS,das, Dokuchaev2012}.

Steady accreting flows constitute a simple but nontrivial  accretion model; this allows one to see effects of backreaction---or selfgravity of fluids---we use words selfgravity or backreaction interchangeably. They have been discovered and partly analysed in \cite{PRD2006,KM2006,KMRS,Dokuchaev2012}. One of interesting features is that there are some characteristics of accretion critical flows that do not depend on  backreaction, assuming certain boundary conditions. While the areal radius of the sonic sphere $R_\ast$ and the mass $m_\ast$ within the sonic sphere might depend on selfgravity, their ratio $m_\ast/R_\ast$, the sonic speed of sound $a_\ast$, and the infall velocity $U_\ast$ might be the same as in the critical accretion of massless fluids in a Schwarzschild spacetime. We shall say that these characteristcs---$m_\ast/R_\ast $, $a_\ast$ and $U_\ast$ are intensive. There has been in recent years a continuing interest in the general relativistic accretion, with different aspects investigated by several researchers \cite{Miller, Ortiz, Kremer}.

There are two goals of this paper. First, we extend our former results into a wider class of polytropic equations of state, including those formulated in terms of the baryonic mass density. We find boundary conditions that allow us to provide analytic proofs to the aforementioned observation. It appears that intensive sonic point parameters do not depend on selfgravity and that the profile of the mass accretion rate depends on the mass fraction of  gas in an accreting system. This result is shown to be true for a wide class of polytropes with a polytropic index $\Gamma \in (1,2]$. Second, we demonstrate by finding numerical examples, that the boundary conditions established beneath play a significant role---if they are not satisfied, then all characteristics of sonic points of critical flows depend on selfgravity.  

The order of this paper is as follows. In the next two sections we present equations describing selfgravitating accreting flows, assuming approximately steady accretion. We would like to stress that there emerge two measures of the mass accretion rate. Only one of them---a quantity denoted later on as $\dot m$---is directly related to the asymptotic mass of the whole configuration. The other, $\dot m_\mathrm{B}$, refers to the infall of the baryonic mass. Section IV is dedicated to the description of critical flows, which are characterized by the existence of a critical point---a sonic sphere (sonic point). 

It appears, from the formulae describing parameters of critical points, that characteristics of  sonic points might potentially depend on selfgravity of infalling  fluids; thus they would potentially differ from those describing test (massless) fluids that satisfy the same boundary conditions. Sections V and VI are devoted to the elucidation of this question. We show that under suitable boundary conditions involving the total mass $m$, the external size of the system $R_\infty$ and the asymptotic speed of sound $a_\infty $, intensive characteristics of sonic points are in fact the same as for the critical accretion of test fluids. In Section V we consider a polytropic equation of state $p = K \varrho^\Gamma$ formulated in terms of the rest mass density $\varrho$. This $p$--$\varrho$ equation of state is less popular in the astrophysical literature, but we shall see that the rest mass $\varrho$ directly enters the formula that describes the growth of the mass of a black hole due to the accretion. There are two main results; one of them (Theorem 4) is an explicit realisation of a scenario described in \cite{PRD2006} and it is valid only for $\Gamma \in [1,5/3]$. The other, Theorem 3, demands a stronger boundary condition but it is less restrictive in the choice of positions of sonic points and works for $\Gamma \in (1,2]$. Section VI deals with polytropic equations of state $p = K n^\Gamma$, $\Gamma \in (1,2]$, expressed in terms of the baryonic mass density $n$. Main results are formulated in Theorem 7.

Section VII presents results of numerical investigation. They confirm the validity of theoretical results of Sections V and VI. Not surprisingly, we have found that if the boundary conditions of Theorems 3 and 7 are not satisfied, then all parameters of sonic points depend on selfgravity. This fact has a consequence that is discussed in section VIII; the recently announced analysis of lifetimes of primordial black holes within neutron stars \cite{Baumgarte_Shapiro, BS} hypothesizes that   ``intensive parameters'' of the related accretion are not sensitive to the backreaction. We have found  numerical counterexamples to this assumption. 

Finally, the last Section brings a short summary.

\section{Initial data, quasilocal energy and mass accretion rate}

We will consider a spherically symmetric compact ball of a fluid steadily falling onto a non-rotating black hole. This system has been investigated in detail in \cite{malec,PRD2006}; in this section we briefly recapitulate the main points. 

We will use the line element
\begin{equation}
ds^2 = -N^2 dt^2 + \alpha dr^2 + R^2 \left( d \vartheta^2 + \sin^2 \vartheta d\varphi^2 \right),
\label{ds}
\end{equation}
where the lapse $N$, $\alpha$ and the areal radius $R$ are functions of a coordinate radius $r$ and an asymptotic time variable $t$. The nonzero components of the extrinsic curvature $K_{ij}$ of the $t = \mathrm{const}$ slices read \begin{subequations}
\begin{eqnarray}
K_r^r & = & \frac{1}{2N\alpha} \partial_t \alpha, \\
K_\varphi^\varphi & = & K_\vartheta^\vartheta = \frac{\partial_t R}{NR} = \frac{1}{2} ( \mathrm{tr} K - K_r^r).
\end{eqnarray}
\end{subequations}
Here
\begin{equation}
\mathrm{tr} K = \frac{1}{N} \partial_t \ln \left( \sqrt{\alpha} R^2 \right)
\end{equation}
denotes the trace of the extrinsic curvature $K_{ij}$. The mean curvature of a two-sphere of constant radius $r$, embedded in a Cauchy hypersurface, can be expressed as
\begin{equation}
k = \frac{\partial_r R}{\sqrt{\alpha}}.
\end{equation}
We assume the energy-momentum tensor of the  perfect fluid
\begin{equation}
T^{\mu\nu} = (p + \varrho) u^\mu u^\nu + p g^{\mu\nu},
\end{equation}
where $u^\mu$ denotes the four-velocity of the fluid, $p$ is the pressure and $\varrho$ the energy density. The four-velocity of the fluid is normalized, $u_\mu u^\mu = -1$. Since we only consider radial flows, we have $u_\varphi =u_\vartheta =0$.  We shall also define the energy flux density through a coordinate sphere
\begin{equation}
j \equiv \frac{N}{\sqrt{\alpha }} T^0_r.   
\end{equation}
 
Let us define two optical scalars
\begin{subequations}
\begin{eqnarray}
\theta & \equiv & \frac{2k}{R} + \mathrm{tr} K - K_r^r, \\
\theta^\prime & \equiv & \frac{2k}{R}-(\mathrm{tr} K - K_r^r).
\end{eqnarray}
\end{subequations}
Let $m$ be the asymptotic mass and $S=4\pi R^2$ an area of a two sphere labelled by $R$. The following equality holds true (see \cite{emnom} and Eq.\ (46) in \cite{emnompen}):
\begin{equation}
m = \left(\frac{S}{16\pi}\right)^{1/2}-\frac{R^3}{8} \theta(S) \theta^\prime(S) + m_\mathrm{ext}(r,t),
\label{penrose1}
\end{equation}
where  
\begin{equation}
m_\mathrm{ext}(r,t) = \pi \int_r^\infty \sqrt{\alpha}dr^\prime R^3\left[ \rho(\theta +\theta ')
+j (\theta -\theta ')\right]. 
\label{penrose} 
\end{equation}
Equation (\ref{penrose1}) is valid on any asymptotically flat Cauchy hypersurface, assuming a quick enough falloff of matter fields at spatial infinity.  The quantity $m_\mathrm{ext}(r,t)$ can be interpreted as the quasilocal  contribution to the asymptotic mass $m$ coming from the mass distribution outside the coordinate sphere $r$. The quasilocal mass $m_\mathrm{ext}(r,t)$ changes in time according to the following formula  
\begin{equation}
\partial_t m_\mathrm{ext} = 
4 \pi R^2 k j + 2\pi N R^3 \left(\mathrm{tr} K  - K_r^r\right) p.
\label{tempo}
\end{equation}
We introduce an auxiliary quantity
\begin{eqnarray}
\label{i9}
\tau & = & \frac{3}{4R}\int_R^\infty  {R^\prime}^2(K_r^r)^2 d R^\prime - \\
&& \frac{1}{4R}
\int_R^\infty {R^\prime}^2(\mathrm{tr} K)^2 d R^\prime - \frac{1}{2R} \int_R^\infty \mathrm{tr} K
K_r^r {R^\prime}^2 d R^\prime. \nonumber
\end{eqnarray}
Imposing a integral slicing condition
\begin{equation}
\label{f1}
\tau = \left[ \frac{R (\mathrm{tr} K - K_r^r)}{2}\right]^2,
\end{equation}
one gets, from the momentum constraint equation, $T^0_r = 0$ \cite{malec}, that is $u_r=0$ and $j=0$. Thus all spatial components of the 4-velocity do vanish. The line element (\ref{ds}) with the slicing condition (\ref{f1}) corresponds simply to comoving coordinates. 
 
The mass $m_\mathrm{ext}(r,t)$ becomes now 
\begin{eqnarray}
m_\mathrm{ext}(r,t) & = &
\pi \int_r^{\infty }\sqrt{\alpha}dr^\prime R^3  \rho (\theta +\theta^\prime)\nonumber\\
& = &4\pi \int_r^{\infty }\sqrt{\alpha}dr^\prime R^2  \rho k. 
\label{penrose2} 
\end{eqnarray}
The areal velocity $U$ of a comoving particle designated by coordinates $(r,t)$ is given by
\begin{equation}
U(r,t) = \frac{\partial_t R}{N} = \frac{R}{2}(\mathrm{tr} K - K_r^r).
\end{equation}
From the Einstein constraint equations one has \cite{malec}
\begin{equation}
k =  \sqrt{1 - \frac{2m(R)}{R} + U^2},
\label{p}
\end{equation}
where $m(R)=m-m_\mathrm{ext}(R,t)$ is the mass within the areal sphere of a radius $R$, and
\begin{equation}
\partial_R (R^2 U) - R^2 \mathrm{tr} K = 0.
\label{pp}
\end{equation}

We use here and in what follows the relation $\partial_r = \sqrt{\alpha}k\partial_R$ in order to eliminate the differentiation with respect the comoving radius $r$. The quasilocal mass $m(R)$ obeys the equation
\begin{equation}
\label{masseq}
\partial_R m(R) = 4 \pi R^2 \varrho.   
\end{equation}

We shall define the mass accretion rate---the mass transfer in a unit of time through a sphere of a fixed areal radius $R$---by the formula 
\begin{equation}
 \dot m = -(\partial_t - (\partial_t R)\partial_R) m_\mathrm{ext}(R) =
- 4 \pi NR^2U (\varrho + p).
\label{dot1m}
\end{equation}
It should be emphasized that the mass accretion rate $\dot m$ is related to the asymptotic mass $m$. In fact, we have
\begin{equation}
m=m(R)+m_\mathrm{ext}(R(r),t);
\end{equation}
$m$ is conserved, and thus $(\partial_t - (\partial_t R)m(R) =\dot m$. The mass in the interior of a sphere  of an areal radius $R$ increases by the amount $\dot m$ lost by the exterior of this sphere. In particular, if one considers the accretion onto a black hole, then it is $\dot m$ that adds to its mass content, in a unit of time. 

\section{Steady flow approximation}
 
A collapsing fluid is called steady, if all its characteristics are constant at a fixed areal radius $R$---we follow the definition of Courant and Friedrichs \cite{Courant}. In analytical terms
\begin{equation}
    \partial_t X|_{R = \mathrm{const}} = (\partial_t - (\partial_t R) \partial_R) X = 0,
    \label{s}
\end{equation}
where $X = \varrho$, $U$, $p$ \cite{Misner}. Accretion leads to the change of geometric quantities, like the mean curvature $k$ or the area of the black hole horizon. One can see from formula (\ref{dot1m}), that these changes are negligible for time intervals $\Delta t$ such that $\dot m \Delta t \ll m$,  and on these time scales the notion of steady fluid collapse is justified. 
 
For steady flows evolution equations become just ordinary differential equations.  The Einstein evolution equation
\begin{equation}
 \partial_t U = k^2   \partial_RN  - \frac{m(R)N}{R^2} - 4 \pi N R   
\end{equation}
becomes, taking into account condition (\ref{s}),
\begin{equation}
 (\partial_tR) U = k^2   \partial_RN  - \frac{m(R)N}{R^2} - 4 \pi N R.   
\end{equation} 
This in turn can be written, applying initial constraints (\ref{p}) and (\ref{pp}), as
\begin{equation}
    \partial_R \ln \left( \frac{N}{k} \right) = \frac{4 \pi R (\varrho + p)}{k^2}
\label{N/k}
\end{equation}
(see \cite{malec} for details). In a similar vein, the energy conservation equation
\begin{equation}
\partial_t \varrho = -N \mathrm{tr} K (\varrho +   p) 
\end{equation}
becomes an ordinary differential equation with respect $R$, assuming steady flow. One can show that $\partial_R \dot m = 0$ \cite{malec}.

We will assume that the radius $R_\infty$ of the ball of fluid and boundary data are such that $|U_{\infty}| \ll \frac{m(R_\infty)}{R_\infty} \ll a_\infty$. These boundary conditions are needed in order to glue the steady fluid with an external Schwarzschild geometry (see a discussion in \cite{malec}).
Strictly saying, we demand that our
Cauchy slices end at spatial infinity. Later on we shall present more precise boundary conditions, that guarantee  the existence of intensive parameters of critical flows.

Let $N_\infty$ and $k_\infty$ denote the values of the lapse $N$ and the mean curvature $k$ at the outer boundary $R = R_\infty$. Equation (\ref{N/k}) can be integrated inwards, yielding
\begin{equation}
N = \frac{N_\infty}{k_\infty} \sqrt{1-\frac{2m(R)}{R}+U^2}e^{-4\pi\int_R^{R_\infty} (  p + \varrho) \frac{s}{k^2} ds}.
\label{25}
\end{equation}

We define a black hole as a region within an apparent horizon to the future, i.e., a region enclosed by an outermost sphere $S_A$ on which the optical scalar $\theta \equiv  (2k  + 2U)/R$ vanishes \cite{Wald,nom}. The other optical scalar, $\theta^\prime \equiv  (2k  - 2U)/R,$ is strictly positive for all spheres outside $S_A$, because the velocity is directed inward, $U<0$. The ratio $2m(R_\mathrm{BH})/R_\mathrm{BH}$ becomes 1 on $S_A$, where $R_\mathrm{BH}$ is the areal radius of the apparent horizon.

The momentum conservation equation $\nabla_\mu T^\mu_r = 0$ reads 
 \begin{equation}
\label{euler}
N \partial_R   p + (  p + \varrho) \partial_R N = 0.
\end{equation}
For barotropic fluids with an equation of state of the form $p = p(\varrho)$, Eq.\ (\ref{euler}) can be integrated, yielding
\begin{equation}
\label{bernoulligeneral}
N = C \frac{n}{p + \varrho},
\end{equation}
where $C$ is an integration constant, and
\begin{equation}
\label{baryonicn}
n = n_1 \exp \int_{\varrho_1}^\varrho \frac{ds}{s + p(s)}
\end{equation}
denotes the baryonic (rest-mass) density ($n_1$ is the baryonic density corresponding to a reference energy density $\rho_1$). For perfect fluids with a barotropic equation of state expression (\ref{baryonicn}) works as an integration factor; if $T_{\mu\nu}$ is conserved, then the baryonic current is also conserved, i.e., $\nabla_\mu \left( n u^\mu \right) = 0$. For our model this implies that
\begin{equation}
\partial_R\left( nUR^2\right)=0.
\end{equation}
One can define the baryonic mass accretion rate
\begin{equation}
\dot m_\mathrm{B}\equiv -4\pi nUR^2
\end{equation}
Clearly, in the present context the two mass accretion rates, $\dot m$ and $\dot m_\mathrm{B}$, are proportional, but we do not see any \textit{a priori} argument that would allow us to determine the proportionality constant while remaining within the present model. Additional information would be needed.

If the spacetime is stationary, both accretion rates can be defined on an equal footing, by referring to Killing vectors. In our model, choosing the areal radius as the radial coordinate, we can write the spacetime line element as
\begin{eqnarray}
ds^2 & = & - \left( N^2 - \frac{N^2 U^2}{k^2} \right) dt^2 - \frac{2 N U}{k^2} dt dR + \frac{1}{k^2} dR^2 \nonumber \\
& & + R^2 (d \vartheta^2 + \sin^2 \vartheta d \varphi^2).
\label{4.19}
\end{eqnarray}
For steady flows, all metric functions in Eq.\ (\ref{4.19}) depend only on $R$. Such metrics admit a Killing vector of the form $\xi^\mu = (1,0,0,0)$---they are stationary in the region in which $\xi_\mu \xi^\mu < 0$. As a consequence, the vector $j_\varepsilon^\mu = - T^\mu_\nu \xi^\nu = - T^\mu_t$ satisfies the conservation law $\nabla_\mu j_\varepsilon^\mu = 0$---this follows immediately from the conservation law for the energy-momentum tensor and the Killing equation $\nabla_\mu \xi_\nu + \nabla_\nu \xi_\mu = 0$. It turns out that $\dot m = - 4 \pi R^2 (\rho + p) N U$ can also be expressed as
\begin{equation}
\dot m = \int_0^{2\pi} d \varphi \int_0^\pi d \theta R^2 \sin \theta \, T^R_t,
\end{equation}
i.e., as a (minus) flux of the vector $j_\varepsilon^\mu$ through a sphere of constant radius. Indeed, a straightforward calculation yields for the energy-momentum component $T^R_t = - (\rho+p) N U$. Thus, for stationary metrics, both accretion rates $\dot m$ and $\dot m_\mathrm{B}$ are related to a conserved current. On the other hand, we would like to stress again that it is the mass flux $\dot m$ that is connected to the asymptotic mass, as explained at the end of the previous section. Note that the ratio
\begin{equation}
\frac{\dot m}{\dot m_\mathrm{B}} = \frac{\rho + p}{n} N
\end{equation}
is constant, as required by the Bernoulli equation.

We shall specialize in   Sections IV and V  to polytropic perfect fluids $p = K \varrho^\Gamma$, with $K$ and $\Gamma$ being constant ($1 \le \Gamma \le 2$). Thence Eq. (\ref{euler}) can be integrated, yielding 
\begin{equation}
a^2 = - \Gamma + (\Gamma +a^2_\infty) \left( \frac{N_\infty}{N} \right)^\kappa.
\label{bernoulli}
\end{equation}
Here $\kappa = (\Gamma - 1)/\Gamma$ and $a=\sqrt{d p/d \varrho}$ is the speed of sound. Equation (\ref{bernoulli}) can be regarded as the general-relativistic version of the Bernoulli equation.

In Section VI we discuss polytropes $p=Kn^\Gamma$, with $1 \le \Gamma \le 2$, which are more common in astrophysical literature. In this case we have 
\begin{equation}
N\left( \Gamma - 1 - a^2_\infty \right)=N_\infty \left( \Gamma - 1 - a^2 \right),
\label{nnB}
\end{equation}
where $N_\infty =N(R_\infty)$. In the following, we will refer to both polytropic equations of state as the $p$--$\varrho$ and $p$--$n$ case, respectively.

Stationary accretion of test fluids, i.e., on a Schwarzschild background, has been studied extensively for both types of polytropic equations of state. Accretion with the equation of state $p = K\varrho^\Gamma$ has been investigated in \cite{malec}. The model with the equation of state $p = K n^\Gamma$ has been studied in \cite{michel, begelman, CS, CMS}. For the test-fluid accretion with the equation of state $p = Kn^\Gamma$ Chaverra et al.\ \cite{CMS} have shown that the ratio $L \equiv \dot m/\dot m_\mathrm{B}$ plays an important role regarding the existence and the character of solutions. For $\Gamma > 5/3$, if $L$ is larger than 1, only global critical solutions can exist. But if $L$ is smaller than 1, but larger than a particular threshold number, then there appears a new, homoclinic, family of solutions.

\section{Critical flows and mass accretion rate}

We shall focus our attention on the so-called critical accreting flows.  There exists a point---the sonic point---at which  
the length of the spatial velocity vector $|{\vec U}| = |U|/k)$ equals $a$. Let us point that this sonic point definition coincides, for test fluids, with the one used by Shapiro and Teukolsky $U^2 = a^2/(1 + 3 a^2)$ \cite{shapiro_teukolsky}. In the Newtonian limit this coincides with the well known requirement $|U|=a$. In the external layers of gas $a>|U|/k$, while in the interior of the sonic sphere $a<|U|/k$. 
  
Critical accreting flows are of interest, because there exist arguments (albeit at present they concern only newtonian hydrodynamical accretion), that they maximize the mass acretion rate  \cite{KMRS, Pad}. 
  
In the following we will denote with an asterisk all values referring to the sonic points, e.g., $a_\ast$, $U_\ast $, etc. The four characteristics, $a_\ast$, $U_\ast$, $m_\ast / R_\ast$, and $c_\ast$ are related \cite{malec},
\begin{equation}
a_\ast^2 \left( 1 - \frac{3m_\ast}{2R_\ast} + c_\ast \right) = U_\ast^2 = \frac{m_\ast}{2R_\ast} + c_\ast,
\end{equation}
where $c_\ast = 2 \pi R^2_\ast   p_\ast$. 

The following lemma bounding the value of $a_\ast^2$ is valid for general barotropic equations of state and in particular for both polytropes discussed in this paper: $p = K \varrho^\Gamma$ and $p = K n^\Gamma$. For the polytropic equation of state $p = K n^\Gamma$ we have a stronger, general bound: $a^2 \le \Gamma - 1$.

\begin{lem}
If $R_\ast > 2m_\ast$, that is a sonic point exists outside of an apparent horizon, then $a^2_\ast \le 1$.
\end{lem}

\begin{proof}
One gets from the definition of the sonic point
\begin{equation}
\frac{m_\ast}{2R_\ast}=\frac{-c_\ast+a^2_\ast +a^2_\ast c_\ast }{1+3a^2_\ast}.
\label{lemma1}
\end{equation}
Assume the opposite---that $a^2_\ast >1$. Then one obtains that the right hand side of (\ref{lemma1})  is bigger than $1/4$, which would imply $R_\ast<2m_\ast$, in contradiction to the assumption.
\end{proof}

The mass accretion rate (\ref{dot1m}) can be expressed by characteristics of
the sonic point \cite{malec}:
\begin{equation}
\dot m = - 4\pi R^2_\ast  \left( 1 - \frac{3m_\ast}{2 R_\ast} + c_\ast \right)^{1/2} U_\ast \varrho_\ast  \left(1 + \frac{p_\ast}{\varrho_\ast} \right) \frac{N_\ast}{k_\ast}.
\label{4.18}
\end{equation}
It was found numerically in \cite{PRD2006} that for $\Gamma <5/3$ and under suitable and quite subtle boundary conditions, the constant $c_\ast $ is negligible. A scenario for the analytic proof, (again for $\Gamma <5/3$), has been also outlined   \cite{PRD2006}.
We shall prove in the next two Sections various generalizations  of this statement; most of them are valid  for $\Gamma \in (1,2]$.  Assuming that $c_\ast \approx 0$, we would have 
\begin{equation}
  U_\ast^2 =
   \frac{m_\ast}{2R_\ast}  = a_\ast^2 \left( 1 - \frac{3m_\ast}{2R_\ast}   \right).
   \label{crit} 
   \end{equation}
These formulas are similar to those valid for test fluids---the only difference is that the asymptotic mass $m$ is being replaced by the quasilocal mass $m_\ast $ within the sonic sphere. In Section \ref{pnproofs} we shall prove an analogous result for the equation of state $p = K n^\Gamma$. It appears, that---similarly to the test fluid approximation---the parameters $U_\ast$, $a_\ast$ and $m_\ast/R_\ast$ of the sonic point do not depend on the central mass.

We shall demonstrate that under suitable boundary conditions the lapse $N$ is well approximated by the mean curvature $k$, in the equation for the sonic point. Consequently, an estimate on $k$ will play and important role in our reasoning.

We will always demand that at the boundary of the cloud of accreting gas $U^2_\infty \ll m/R_\infty \ll a^2_\infty $; these are expected characteristics of the external branch of a critical accreting solution. If $m/R_\infty$ is sufficiently small, then
\begin{equation}
    k_\infty = \sqrt{1 - \frac{2m}{R_\infty} + U_\infty^2} \approx 1,
\end{equation}
a condition, which we will also assume in the analytic part of this paper. In this case, Eq.\ (\ref{25}) reads
\begin{equation}
N = k e^{-4\pi\int_R^{R_\infty} (  p + \varrho) \frac{sds}{k^2 }},
\label{lapse}
\end{equation}
where we have additionally set $N_\infty = 1$. Setting $N_\infty = 1$ is another simplifying convention, which we adopt in the remainder of this paper. It does not alter the reasoning, but saves a lot of writing.

Also, in all analytic results, we assume that the density $\rho$ (or, equivalently $n$) is a non increasing function of the radius $R$. This assumption excludes homoclinic-type solutions discovered for the test-fluid case in \cite{CS} and analyzed in some detail in \cite{CMS}. In general, there is no reason preventing the existence of homoclinic-type solutions (characterized by a non monotonic dependence of the density on the radius $R$) in the full self-gravitating case, and indeed, in Sec.\ \ref{numerics} we give a numerical example of such a solution.

\section{Sonic points for the equation of state 
$p=K\varrho^\Gamma .$}

The infall velocity $U$ reads \cite{malec}
\begin{equation}
U=U_\ast \frac{R^2_\ast}{R^2}
\left(  \frac{1 + \frac{\Gamma}{a^2}}{1 + \frac{\Gamma}{a^2_\ast} } \right)^\frac{1}{\Gamma - 1}.
\label{U}
\end{equation}
Here $U_\ast$ is the negative square root. From the relation between the pressure and the
energy density, one obtains, using  equation (\ref{bernoulli})
\begin{equation}
\varrho = \varrho_{\infty } \left( \frac{a}{a_\infty} \right)^\frac{2}{\Gamma - 1} = \varrho_\infty
\left( - \frac{\Gamma}{a_\infty^2} + \frac{\frac{\Gamma}{a_\infty^2}  + 1}{N^\kappa} \right)^ \frac{1}{\Gamma - 1},
\label{rho}
\end{equation}
where the constant  $\varrho_\infty$ is equal to the mass density of a collapsing fluid at the boundary $R_\infty$. The steady fluid is described by Eqs.\ (\ref{bernoulli}), (\ref{U}), and (\ref{rho}). They constitute an integro-algebraic system of equations,
with a bifurcation point at the sonic point, where two branches (identified as accretion
or wind) do cross; that is a feature   present also in models with a test fluid \cite{bondi}--\cite{malec}. That requires some caution in doing numerics and a careful selection of the solution branch.  

Using (\ref{rho}) and $a^2 = \Gamma p/\varrho$, we immediately get  the mass accretion rate adapted to the equation of state $p=K\varrho^\Gamma $ \cite{malec}:
\begin{eqnarray}
\dot m & = & - 4 \pi R^2_\ast \rho_\infty \left( 1 - \frac{3m_\ast}{2 R_\ast} + c_\ast \right)^{1/2} \nonumber \\
& & \times U_\ast \left( \frac{a_\ast}{a_\infty}
\right)^\frac{2}{\Gamma - 1} \left(1 + \frac{a^2_\ast}{\Gamma} \right) \frac{N_\ast}{k_\ast}.
\label{4.18a}
\end{eqnarray}

It will be proven below that, under suitable conditions, significant information about a full system with backreaction can be obtained by investigating steady flows with the backreaction being ignored. The main results will be formulated in the two forthcoming subsections. They demand different boundary conditions. The boundary conditions $U^2_\infty \ll m/R_\infty \ll a^2_\infty $ are required to prove Theorem 4, using a proof scenario outlined in \cite{PRD2006}. In addition we consider $\Gamma \in [1,5/3]$ and assume that the sonic radius $R_\ast $ is not smaller than $22m/3$.  

In Lemma 2 we demand that $2400^{(\Gamma-1)/\Gamma}\left( m/R_\infty \right)^{3(\Gamma -1)/\Gamma} \ll a_\infty^2 $; this might be weaker than before for $2\ge \Gamma \ge 3/2$, but it is a stronger condition than the former one, for $\Gamma < 3/2$. Yet stronger condition (for any exponent $\Gamma$,  $1<\Gamma \le  2$) $\left( 6m/R_\infty \right)^{(\Gamma -1)} \ll a_\infty^2 $ is imposed in the proof of Theorem 3, but this allows one to have $R_\ast$ as small as one wishes to have, in contrast to Theorem 4.  Boundary conditions in Lemma 6 and Theorem 7 are more complex.

\subsection{A bound on $k$}

\begin{lem}
\begin{enumerate}[i)]
\item Assume an accreting solution, with a monotonically decreasing mass density $\varrho $, $\Gamma \in (1, 2]$, and asymptotic data $\varrho_\infty, R_\infty, m $ and $a_\infty$.
\item Assume that $R_\infty \gg m $.
\item Let there exists a sonic point at a radius $R_\ast $ and mass $m_\ast$.
\end{enumerate}
Then for $R\ge R_\ast$,
\begin{equation}
k\ge \inf \left\{ \sqrt{0.9}, \left( 1-\frac{2m_\ast}{R_\ast} -\frac{2400m^2\left( m-m_\ast\right)}{R_\infty^3a_\infty^{2/(\Gamma -1)}}\right)^{1/2} \right\}.
\label{boundk}
\end{equation}
If in addition $2400^{\frac{\Gamma -1}{\Gamma}} \left(\frac{m}{R_\infty }\right)^{\frac{3\left(\Gamma -1\right)}{\Gamma}} \ll a_\infty^2$,
then 
\begin{equation}
k\ge \inf \left\{ \sqrt{0.9}, \left( 1-\frac{2m_\ast}{R_\ast}   \right)^{1/2} \right\}.
\label{boundkklemma2}
\end{equation}
\end{lem}

\begin{proof}
The speed of sound is smaller than 1 outside a black hole (see Lemma 1). Thus we can bound the mass density, using (\ref{rho}):
\begin{equation}
\varrho \le \varrho_\infty \left( \frac{1}{a^2_\infty}\right)^{\frac{1}{\Gamma -1}}.
\label{boundrho1}
\end{equation}
The mass function $m(R)$ can be written as
\begin{equation}
m(R)=m_\ast +4\pi \int_{R_\ast}^R\varrho r^2dr.
\end{equation}
Thus one obtains, using (\ref{boundrho1}) and an obvious inequality
\begin{equation}
\frac{4\pi}{3} \varrho_\infty \left(R_\infty^3-R_\ast^3\right)\le m-m_\ast
\label{obvious}
\end{equation}
(note that $\varrho \ge \varrho_\infty$), the following estimate:
\begin{equation}
m(R)\le m(R_\ast)+\frac{m-m_\ast}{a_\infty^{2/(\Gamma -1)}\left(R_\infty^3-R_\ast^3\right)}\left( R^3-R_\ast^3\right).
\label{boundrho2}
\end{equation}
Let us   first estimate $k$ in the interval   $(R_\ast,20m)$. We need to estimate  the mass function in the interval 
 $(R_\ast, 20m)$.
Approximating $R^3-R_\ast^3=(R-R_\ast)(R^2+RR_\ast+R_\ast^2)\le 3(R-R_\ast)R^2$, we get the bound 
\begin{equation}
m(R)\le m(R_\ast)+\frac{\left(m-m_\ast\right) \left( R-R_\ast\right) }{a_\infty^{2/(\Gamma -1)}R_\infty^3 }3\left( 20m  \right)^2.
\label{boundrho3}
\end{equation}
 This implies the following bound onto the function $k(R)$ for $R\in (R_\ast, 20m)$:
\begin{eqnarray}
k(R)& = & \sqrt{1-2\frac{m(R)}{R}+U^2} \nonumber\\
&\ge& \sqrt{1-2\frac{m_\ast}{R_\ast}-\frac{\left(m-m_\ast\right)   }{a_\infty^{2/(\Gamma -1)}R_\infty^3 }  2400m^2  }\nonumber\\
& \ge & \sqrt{1-2\frac{m_\ast}{R_\ast}-\frac{ 2400m^3 }{a_\infty^{2/(\Gamma -1)}R_\infty^3}}.
\label{boundk1}
\end{eqnarray}
This becomes
\begin{equation}
k(R)\ge \sqrt{1-2\frac{m_\ast}{R_\ast}}, 
\label{boundk2}
\end{equation}
if the third term in the square root in (\ref{boundk1}) is much smaller than $a^2_\infty$ (beware that $m_\ast/R_\ast$ is not smaller than $a^2_\infty $), that is if  
\begin{equation}
2400^{\frac{\Gamma -1}{\Gamma }} \left(\frac{ m}{R_\infty }\right)^{\frac{3\left(\Gamma -1\right)}{\Gamma }}\ll a_\infty^2.
\label{est1}
\end{equation} 
For $R \ge 20m$, and thus also for $R \ge R_\ast \ge 20m$, we get 
\begin{equation}
k(R)\ge \sqrt{1-\frac{2m}{20m}}= \sqrt{0.9}
\label{est2}
\end{equation}
---the first term in bound (\ref{boundkklemma2}). This accomplishes the proof of Lemma 2.
\end{proof}

\begin{remark}
In the case of $\Gamma =5/3$ we get from (\ref{est1}) the bound $a_\infty^2 \gg 25(m/R_\infty)^{1.2}$.
\end{remark}

\subsection{Main result}

 We assume that there are  no apparent horizons at the immediate vicinity of a sonic horizon---we put $R_\ast \ge 4m_\ast$---but otherwise there are no restrictions onto $R_\ast $.

\begin{thm}
 Let the mass density $\varrho $ be monotonically decreasing.
\begin{enumerate}[i)]
\item Assume $\Gamma \in (1, 2]$, and   asymptotic data
 $\varrho_\infty, R_\infty, m $ and $a_\infty$.
\item Assume $R_\infty \gg m $ and $6^{(\Gamma -1)/2}  \left(\frac{m}{R_\infty }\right)^{\frac{\left(\Gamma -1\right)}{2}}\ll a_\infty$.
\item Let there exists a sonic point  at $R_\ast >4m_\ast$.
\end{enumerate}
Then the  sonic point parameters $a^2_\ast$, $U^2_\ast$, and $m_\ast/R_\ast$ in the above model with backreaction are essentially the same as in the test fluid  accretion with the same asymptotic data.
\end{thm}  
  
\begin{remark}
Notice, that condition ii) is stronger than the similar assumption of Lemma 2 (see also (\ref{est1}) for a comparison).
\end{remark}
   
\begin{proof}
Condition ii) allows us to use estimate (\ref{boundk2}) of Lemma 2:
\begin{equation}
k \ge k_\mathrm{inf} \equiv \inf \left\{ \sqrt{0.9}, \left( 1-\frac{2m_\ast}{r_\ast}   \right)^{1/2} \right\}.
\label{boundkka}
\end{equation}
Condition iii) yields now   $k_\mathrm{inf}=\sqrt{0.5}$. Employing this fact together with the relation $p \le \varrho$ and estimate (\ref{boundrho1}), we obtain the following chain of inequalities:
\begin{eqnarray}
\lefteqn{4\pi\int_{R_\ast}^{R_\infty} (p + \varrho) \frac{s}{k^2} ds \le 4 \pi \frac{1}{0.5} \int_{R_\ast}^{R_\infty} (  p + \varrho) s ds \le } \nonumber \\
&& 16\pi \int_{R_\ast}^{R_\infty}  \varrho_\infty \frac{1}{a_\infty^{\frac{2}{\Gamma-1}}} s ds \le  6 \frac{m - m(R_\ast)}{ a_\infty^{\frac{2}{\Gamma-1}}R_\infty}\ll 1.  
 \label{e11a}
 \end{eqnarray}
In order to get the last inequality, we exploited the inequality $\frac{4\pi}{3} \varrho_\infty \left(R_\infty^3-R_\ast^3\right)\le m-m_\ast$.
 
Inserting (\ref{e11a}) into the expression for the lapse $N$, we 
 get the estimate
\begin{equation}
N(R_\ast)\approx \sqrt{1-2\frac{m_\ast}{R_\ast}+U^2_\ast}.
\label{e2a}
\end{equation}
Inserting that into the Bernoulli-type equation (\ref{bernoulli}), we obtain the equation 

\begin{equation}
a^2_\ast  = - \Gamma + \frac{\Gamma + a^2_\infty}{{\sqrt{1 - 
\frac{2   m_\ast}{R_\ast} +  U^2_\ast}}^\kappa}.
\label{estimatec}
\end{equation}

Using $p = a^2 \varrho/\Gamma$, $a^2_\ast \le 1$ and (\ref{obvious}), one gets
\begin{equation}
c_\ast = \frac{a^2_\ast}{\Gamma} 2 \pi R^2_\ast \varrho_\ast \le   a^2_\ast
\frac{3m}{2 a_\infty^{\frac{2}{\Gamma -1}}} \frac{R_\ast^2}{R_\infty^3} \ll a^2_\ast .
\end{equation}
That implies that $c_\ast$ can be put to zero in all relations between characteristics of the sonic point. Inserting this information into the Bernoulli equation (\ref{bernoulli}), one obtains that at a sonic point \cite{malec}
\begin{equation}
1 + y (3 \Gamma - 1) = 3 \left(a^2_\infty + \Gamma \right) y^\frac{\Gamma + 1}{2 \Gamma},
\label{sonica}
\end{equation}
where $y = 1 - 3 m_\ast / (2R_\ast)$. Coefficients of this algebraic equation do not depend on the asymptotic mass density, and therefore $y$ is independent of $\varrho_\infty$. The sonic mass $m_\ast$ and the sonic radius $R_\ast$ are clearly dependent on $\varrho_\infty$, but their ratio is constant. In fact, $m_\ast/ R_\ast$ must be he same as in the case in which the backreaction can be neglected, that is when the mass of the fluid outside the black hole is small in comparison to the total mass. The same conclusion holds true also for other parameters of the sonic point, the fluid velocity $U_\ast$ and the speed of sound $a_\ast$. In conclusion: $a^2_\ast $, $U^2\ast$ and $m_\ast/R_\ast$ can be inferred from a suitable steady flow with a test fluid.
\end{proof}

If backreaction can be neglected, then the mass accretion rate (\ref{4.18a}) reads \cite{PRD2006}
\begin{eqnarray}
\dot m & = &  \pi m^2_\ast \frac{\rho_\infty}{a^3_\infty } \left( 1 +3a^2_\ast \right) \nonumber \\
& & \times   \left( \frac{a_\ast}{a_\infty}
\right)^\frac{5-3\Gamma }{\Gamma - 1} \left(1 + \frac{a^2_\ast}{\Gamma} \right) .
\label{4.18b}
\end{eqnarray}
This is the same formula as that for the accretion of test fluids---see Eq.\ (6.1) in \cite{malec}---but with the mass $m_\ast$ instead of the asymptotic mass $m$.

\subsection{Sonic points outside $22m/3$.}
\begin{thm}
Let the mass density $\varrho $ be monotonically decreasing.
\begin{enumerate}[i)]
\item Assume $\Gamma \in [1, 5/3]$ and asymptotic data $\varrho_\infty$, $R_\infty$, $m$, and $a_\infty$, and let $U_\infty^2 \ll  \frac{m }{R_\infty} \ll a_\infty^2$.
\item Assume that the sonic point is located outside the apparent horizon, $R_\ast \ge 22 m/3$.
\item Furthermore, let the mass within the sonic sphere be given by $m_\ast =\frac{m}{1+\beta }$, where $\beta >0$.
\item Define
\begin{equation}
D\equiv \frac{3}{4\Gamma}
(\Gamma - 1)^2 (9\Gamma - 7)+3(5 - 3 \Gamma ) .
\label{b1}
\end{equation}
and
\begin{eqnarray}
&& F(\Gamma, a^2_\infty)\equiv \nonumber\\
 &&\left(\frac{2}{5 - 3 \Gamma +\sqrt{\left(5 - 3 \Gamma\right)^2+8a^2_\infty D}} \right)^{\frac{1}{\Gamma-1}} .
\end{eqnarray} 
Assume that
 $R_\infty \gg m\delta$, where $\delta =\frac{33 }{8}F(\Gamma, a^2_\infty)$.
\end{enumerate}
Then the sonic point parameters $a^2_\ast$, $U^2_\ast$, and $m_\ast / R_\ast$ in the above model with backreaction are the same as in the test fluid  accretion with the same asymptotic data \cite{remark}.
\end{thm}

\begin{proof}
Condition ii) implies $k \ge \sqrt{8/11}$. Using this, and employing the relation $p \le \varrho$, we can obtain a chain of inequalities 
\begin{eqnarray}
\lefteqn{4\pi\int_{R_\ast}^{R_\infty} (  p + \varrho) \frac{s}{k^2} ds \le 4\pi \frac{11 }{8} \int_{R_\ast}^{R_\infty} (  p + \varrho) s ds \le}  \nonumber \\
&& 8\pi \frac{11}{8} \int_{R_\ast}^{R_\infty}  \varrho s ds  \le   \frac{11 }{4} \frac{m - m(R_\ast)}{ R_\ast} = 11\beta \frac{m_\ast}{4R_\ast}.
 \label{e1}
 \end{eqnarray}
In order  to get  the last inequality we exploited the condition iii).
 
Inserting (\ref{e1}) into the expression for the lapse $N$, we 
 get the estimate
\begin{equation} 
N(R_\ast)\ge \sqrt{1-2\frac{m_\ast}{R_\ast}+U^2_\ast}
e^{-11 \beta \frac{m_\ast}{4R_\ast}}\label{e2b}
\end{equation}
 or (using $U^2_\ast\ge m_\ast/(2R_\ast)$)
\begin{equation} 
N(R_\ast)\ge \sqrt{1-\frac{3m_\ast}{2R_\ast} }
e^{-11\beta \frac{m_\ast}{4R_\ast}}.
\label{e3}
\end{equation} 
One can show that
\begin{eqnarray}
\lefteqn{\sqrt{1-\frac{3m_\ast}{2R_\ast} }
e^{-11 \beta \frac{m_\ast}{4R_\ast}}\ge} \nonumber\\
&& \sqrt{1-\frac{m_\ast}{R_\ast}\left(2+\frac{22\beta}{3}\right)  + \hat U^2},
\label{e4}
\end{eqnarray}
where
\begin{equation}
\hat U^2=\frac{m_\ast}{4R_\ast}\left( 2 + 22\frac{\beta}{3} \right).
\label{e5}
\end{equation}
We can apply estimate (\ref{e4}) in Eq.\ (\ref{bernoulli}); one arrives at
\begin{equation}
a^2_\ast \le - \Gamma + \frac{\Gamma + a^2_\infty}{\sqrt{ 1-\frac{m_\ast}{R_\ast}\left(2+
22\frac{  \beta}{3 }\right)  +\hat U^2 }^\kappa} .
\label{estimate}
\end{equation}
Solutions of this inequality are bounded from above by the solution of the equation 
\begin{equation}
\hat a^2  = - \Gamma + \frac{\Gamma + a^2_\infty}{{\sqrt{1 - 
\frac{2 \hat m}{R_\ast} + \hat U^2}}^\kappa},
\label{estimatea}
\end{equation}
where $\hat m =  m_\ast \left( 1 + 11 \beta/3 \right)$, at a point where $\hat a / \left(1 + 
3 \hat a^2 \right) = \hat m / (2 R_\ast) = \hat U^2$. 
The case of equality in Eq.\ (\ref{estimate}) is a sonic point equation for the test fluid accretion in a fiducial Schwarzschild geometry
\begin{equation}
ds^2=- \left( 1-\frac{\hat m}{R} \right) dt^2+\frac{dR^2}{1-\frac{\hat m}{R}}+R^2d\Omega^2.
\end{equation}
Note that there is a consistency condition $R_\ast > 2 \hat m $, that is  the sonic point in the fiducial metric is located outside its event horizon. It is satisfied---in this place we have to invoke assumption ii). If ii) is true then we are guaranteed that the sonic point exists outside the event horizon in this fiducial metric. This allows us to use the estimate of Theorem 2 of \cite{malec} that holds true for sonic points located outside event horizons. It appears that
\begin{equation}
 \hat a^2 \le \frac{2a^2_\infty}{5 - 3 \Gamma + \frac{3 \hat a^2 
(\Gamma - 1)^2 (9\Gamma - 7)}{4 \Gamma \left(1 + 3 \hat a^2 \right)}}.
\label{theorem2}
\end{equation}
Since   $a^2_\ast < \hat a^2$, we can conclude that
\begin{equation}
   a^2_\ast \le \frac{2a^2_\infty}{5 - 3 \Gamma + \frac{3 a^2_\ast
(\Gamma - 1)^2 (9\Gamma - 7)}{4 \Gamma \left(1 + 3  a^2_\ast \right)}}.
\label{theorem3}
\end{equation}
This is a bi-quadratic inequality. Standard reasoning yields the following estimate for $a^2_\ast$:
\begin{equation}
a_\ast^2\le \frac{4a^2_\infty}{5 - 3 \Gamma  -6a^2_\infty+\sqrt{\left(5 - 3 \Gamma-6a^2_\infty\right)^2+8a^2_\infty D}} .
\label{bi}
\end{equation}

This implies that in the nonrelativistic limit (understood as the condition $a^2_\infty\ll 5-3\Gamma $) the sonic speed satisfies  $a^2\ast \le 2a^2_\infty /(5-3\Gamma) $;  in fact, we have in this case the equality \cite{malec}.  In the case of 
$\Gamma =5/3$ one obtains  $a^2_\ast \le 1.12 a_\infty $, again in agreement with the former estimate of \cite{malec}.

Monotonically decreasing mass density implies that also  the speed of sound is nonincreasing.
Therefore 
\begin{eqnarray}
   a^2(R)&\le &\frac{4a^2_\infty}{E},
\label{theorem4}
\end{eqnarray}
where $E$ denotes the denominator of (\ref{bi}),
\begin{equation}
E=5 - 3 \Gamma -6a^2_\infty +\sqrt{\left(5 - 3 \Gamma -6a^2_\infty\right)^2+8a^2_\infty D} ,
\label{bi1}
\end{equation}
Inserting this  estimate into the first equality of (\ref{rho}) yields the bound on the matter density,
\begin{equation}
 \varrho(R) \le \varrho_\infty F(\Gamma, a^2_\infty),
 \label{bi2}
 \end{equation}
 where 
 \begin{equation}
 F(\Gamma, a^2_\infty)\equiv \left(\frac{4}{E} \right)^{\frac{1}{\Gamma-1}} .
 \label{theorem5}
 \end{equation}
One can check that the quantity $F$ varies from around 1 (for $\Gamma =1$ and small asymptotic values of the speed of sound) to a multiple of $1/a_\infty^{3/2}$ (for $\Gamma =5/3$). However, it  is smaller than 50, irrespective of $a_\infty $, if $\Gamma <1.6$.

This in turn allows one to replace the bound appearing in Eq. (\ref{e1}),
\begin{equation}
 4\pi\int_{R_\ast}^{R_\infty} ( p + \varrho) \frac{s}{k^2} ds \le 8\pi \frac{11 }{8} \int_{R_\ast}^{R_\infty}  \varrho s ds;
\label{theorem6}
\end{equation} 
using Eq.\ (\ref{theorem5}), we can write 
\begin{equation}
\int_{R_\ast}^{R_\infty}  \varrho sds\le  
\varrho_\infty 
 \frac{F}{2} ( R^2_\infty  -R^2_\ast);
\label{theorem7}
\end{equation}
obviously $\varrho_\infty \frac{1}{2} ( R^2_\infty  -R^2_\ast)\le \frac{m}{R_\infty}\frac{3}{8\pi}$, and inequality (\ref{theorem6}) can be replaced by
\begin{equation}
4\pi\int_R^{R_\infty} (  p + \varrho) \frac{s}{k^2 } ds \le 3F \frac{11}{8} \frac{m}{R_\infty}.  
\label{theorem8}
\end{equation} 
We see, invoking assumption iv), that the lhs of inequality (\ref{theorem8}) is very small and the lapse at the sonic point can be approximated by the mean curvature $k$.

Using the estimates $a^2_\ast <1$, (\ref{bi2}) and $\varrho_\infty \le 3m/(4\pi R^3_\infty )$, one gets
\begin{eqnarray}
c_\ast & = & \frac{a^2_\ast}{\Gamma} 2 \pi R^2_\ast \varrho_\ast \le F a^2_\ast 
\frac{3m}{4\pi R_\infty} \frac{R_\ast^2}{R_\infty^2} \ll \frac{Fm}{R_\infty}.
\end{eqnarray}
That implies that  $c_\ast$ can be put to zero in all relations between characteristics of the sonic point. Inserting 
this information into the Bernoulli equation (\ref{bernoulli}) one obtains that at a sonic point \cite{malec}
\begin{equation}
1 + y (3 \Gamma - 1) = 3 \left(a^2_\infty + \Gamma \right) y^\frac{\Gamma + 1}{2 \Gamma},
\label{sonicb}
\end{equation}
where $y = 1 - 3 m_\ast / (2R_\ast)$. Coefficients of this algebraic equation do not depend on  the asymptotic mass density and therefore $y$ is independent of $\varrho_\infty$. 

In conclusion: $a^2_\ast $, $U^2\ast $
 and $m_\ast / R_\ast$ can be inferred from a suitable steady flow with a test fluid. That ends the proof of Theorem 4.
\end{proof}

We believe that the assumption of Theorem 4,  that $R_\ast \ge 22m/3 $, can be relaxed, as suggested in \cite{PRD2006}.
  
\section{Characteristics of sonic points for the polytropic equation of state $p=Kn^\Gamma .$}
\label{pnproofs}

We shall restrict our attention in this Section to solutions  with the monotonically decreasing baryonic mass density $n$. We assume that the polytropic index $\Gamma$ belongs to the interval $(1,2]$.

\subsection{Estimation of baryonic and rest mass densities.}

\begin{lem}
Consider a spherically symmetric accretion of test fluids with the polytropic equation of state $p = K n^\Gamma$ in a Schwarzschild spacetime of mass $m$. Assume that the density $n$ is a nonincreasing function of the radius $R$. Let $\Gamma \in (1, 2]$, and let the asymptotic value of the speed of sound be $a_\infty$. Then outside the sonic sphere the speed of sound $a^2$ satisfies the bound 
\begin{equation}
a^2< \Gamma-1-e,
\label{an}
\end{equation}
where 
\begin{equation}
e \equiv \frac{\Gamma-1-a^2_\infty}{\sqrt{3\Gamma -2}}.
\label{an1}
\end{equation}
\end{lem}

\begin{proof}
Assuming $N_\infty = 1$, the Eq.\ (\ref{nnB}) reads
\begin{equation}
N=\frac{\Gamma -1-a^2}{\Gamma -1-a^2_\infty}.
\label{n1}
\end{equation}
Since we deal with a vacuum metric, we also have $N = k$---see Eq.\ (\ref{lapse}).

The monotonic falloff of $n$ implies that the speed of sound is monotonically decreasing (see the formula (\ref{nb1}) below). This implies---from Eq.\ (\ref{n1})---that the lapse $N$ is a nondecreasing function of $R$. It is larger than $k(R_\ast)$ which in turn is equal to $\sqrt{1-\frac{3m}{2R_\ast}}$. Notice that  
\begin{equation}
\frac{3m}{2R_\ast}= \frac{3a_\ast^2}{1+3a_\ast^2}
\label{n2}
\end{equation}
---see a discussion preceding Lemma 1. Therefore one obtains from  $N=k$  and Eq.\ (\ref{n2}) that $ N_\ast = \sqrt{\frac{1}{1+3a^2_\ast}}$. Notice that $a^2_\ast<\Gamma-1$; hence $N_\ast> 1/\sqrt{3\Gamma -2}$ and
\begin{equation}
 \frac{\Gamma -1-a^2_\ast}{\Gamma -1-a^2_\infty}> \frac{1}{\sqrt{3\Gamma -2}}.
\label{n3}
\end{equation}
The last inequality can be written as 
\begin{equation}
 a^2_\ast < \Gamma -1-e.
\label{n33}
\end{equation}
The bound in Eq.\ (\ref{an}) follows now directly from the inequality $a^2\le a^2_\ast$ and (\ref{n33}).
\end{proof}

We are now able to give a useful estimate of the baryonic mass density $n$, that is true both for a  test fluid in a Schwarzschild spacetime and for a selfgravitating fluid, provided that they satisfy the same equation of state and   possess the same speed of gas at the sonic point ($a_\ast$) and at the boundary ($a_\infty $). It is known (see, for instance Eq.\ (22) in \cite{KMRS}) that one can express $n$ as
\begin{equation}
n=n_\infty\left( \frac{\frac{\Gamma -1}{a^2_\infty}-1}{\frac{\Gamma -1}{a^2 }-1}\right)^{\frac{1}{\Gamma -1}}.
\label{nb1}
\end{equation}
Applying bound (\ref{an}) on the speed of sound, one arrives at the following estimate of the baryonic mass density
\begin{equation}
n\le n_\infty \left( \frac{\frac{\Gamma -1}{a^2_\infty}-1}{\frac{\Gamma -1}{\Gamma-1-e }-1}\right)^{\frac{1}{\Gamma-1}}.
\label{nb2}
\end{equation}
In the case of the polytropic equation of state $p=Kn^\Gamma$, the rest mass density is given by \cite{KMRS}
\begin{equation}
\varrho=n\left[ 1+\frac{a^2}{\Gamma\left( \Gamma-1-a^2\right)} \right].
\label{nb3}
\end{equation}
Combining Eqs.\ (\ref{nb1}) and (\ref{nb3}), one obtains 
\begin{equation}
\varrho=n_\infty\left( \frac{\frac{\Gamma -1}{a^2_\infty}-1}{\frac{\Gamma -1}{a^2 }-1}\right)^{\frac{1}{\Gamma -1}}\left[ 1+\frac{a^2}{\Gamma\left( \Gamma-1-a^2\right)} \right], 
\label{nb3a}
\end{equation}
and 
\begin{equation}
\varrho_\infty \equiv n_\infty \left[ 1+\frac{a^2_\infty}{\Gamma\left( \Gamma-1-a^2_\infty\right)} \right].
\label{nb3b}
\end{equation}

The use of estimates (\ref{an1}), (\ref{nb2}), and (\ref{nb3})
yields the forthcoming bound onto the rest mass density, outside of the sonic sphere, 
\begin{equation}
\varrho \le n_\infty \left( \frac{\frac{\Gamma -1}{a^2_\infty}-1}{\frac{\Gamma -1}{\Gamma-1-e}-1}\right)^{\frac{1}{\Gamma-1}} \left( 1+ \frac{\Gamma-1-e}{\Gamma e} \right).
\label{nb4}
\end{equation}
This is an important inequality, that will be used later on.
 
To shorten the notation, we define
\begin{equation}
    \eta \equiv \left[ 1+\frac{a^2_\infty}{\Gamma\left( \Gamma-1-a^2_\infty\right)} \right]
\end{equation}
and
\begin{equation}
\Delta\equiv \left( \frac{\frac{\Gamma -1}{a^2_\infty}-1}{\frac{\Gamma -1}{\Gamma-1-e }-1}\right)^{\frac{1}{\Gamma-1}}  \left( 1+ \frac{\Gamma-1-e}{\Gamma e} \right).
\label{nb5}
\end{equation}
Thus Eqs.\ (\ref{nb3b}) and (\ref{nb4}) can be written as $\rho_\infty = n_\infty \eta$ and $\varrho \le n_\infty \Delta$, respectively.
 
 \subsection{A bound on $k$}

\begin{lem}
\begin{enumerate}[i)]
\item Assume an accreting selfgravitating gas in a spherically symmetric spacetime   with $   \Gamma \in (1, 2]$ and  asymptotic data $n_\infty$, $R_\infty$, $m$, and $a_\infty$. Assume that $n$ is a nonincreasing function of the radius $R$.
\item Assume $R_\infty \gg m $.
\item Let there exists a sonic point at a radius $R_\ast $ and mass $m_\ast$, and $a^2_\ast < \Gamma -1-e$.
\end{enumerate}
Then:
\begin{enumerate}[i)]
\item $\frac{2m_\ast}{R_\ast} < \frac{4 (\Gamma - 1 - e)}{3\Gamma - 2 - 3e}$;
\item
\begin{eqnarray}
k & \ge & \inf \left\{ \sqrt{\frac{-\Gamma+2+e}{3\Gamma-2-3e}},\sqrt{0.2}, \right. \nonumber \\
& & \left. \left( 1-\frac{2m_\ast}{R_\ast} -\frac{37.5m^2\left( m-m_\ast\right)}{R_\infty^3 }\frac{ \Delta}{\eta }\right)^{1/2} \right\}.
\label{boundkn}
\end{eqnarray}
\end{enumerate}
If in addition 
\begin{equation}
37.5\left(\frac{m}{R_\infty }\right)^{3}\frac{\Delta}{\eta }\ll a_\infty^2,
\label{lemma6}
\end{equation}
then 
\begin{equation}
k \ge \inf \left\{ \sqrt{\frac{-\Gamma+2+e}{3\Gamma-2-3e}}, \left( 0.2   \right)^{1/2} \right\}.
\label{boundkkn}
\end{equation}
\end{lem}

\begin{proof}
Concerning statement i), we refer to relation (\ref{lemma1}), which implies---taking into account condition iii)---that  at a sonic point
\begin{equation}
\frac{m_\ast}{2R_\ast}\le \frac{a^2_\ast}{1+3a^2_\ast}.
\label{i1}
\end{equation}
 Replacing $a^2_\ast $ by $\Gamma-1-e$  (we  again use assumption iii)), allows us to conclude that 
\begin{equation}
\frac{2m_\ast}{R_\ast } < 4 \frac{\Gamma-1-e}{3\Gamma -2-3e}.
\label{i2}
\end{equation}  

In what follows we deal with part ii) of the lemma. The mass function $m(R)$ can be written as
\begin{equation}
m(R) = m_\ast + 4 \pi \int_{R_\ast}^R \varrho r^2 dr.
\end{equation}
The asymptotic mass density $\varrho_\infty$ is given by Eq.\ (\ref{nb3b}); it can be approximated by the baryonic mass density, if $a_\infty \ll 1$.
 Notice that from (\ref{nb4}) and (\ref{nb5}) we have $\varrho \le n_\infty \Delta$. Thus one obtains, using (\ref{nb4}) and the obvious inequality (\ref{obvious}) (we assume that $n$ is a nonincreasing function of the radius $R$), the following estimate:
\begin{equation}
m(R)\le  m(R_\ast)+\frac{\left(m-m_\ast\right)}{\left(R_\infty^3-R_\ast^3\right)}\left( R^3-R_\ast^3\right)  
\frac{ \Delta}{\eta }.
\label{nb6}
\end{equation}
 
Let us consider first the case with $R_\ast < 2.5m$. We shall estimate the mass function in the interval $(R_\ast, 2.5m)$. Approximating $R^3-R_\ast^3=(R-R_\ast)(R^2+RR_\ast+R_\ast^2)\le 3(R-R_\ast)R^2$, we get
\begin{equation}
m(R)\le m(R_\ast)+\frac{\left(m-m_\ast\right) \left( R-R_\ast\right) }{ R_\infty^3 }3\left( 2.5m  \right)^2\frac{\Delta}{\eta}.
\label{nb7}
\end{equation}
 This implies the following bound onto the function $k(R)$:
\begin{eqnarray}
k(R)&=&\sqrt{1-2\frac{m(R)}{R}+U^2}\ge\nonumber\\ &&\sqrt{1-2\frac{m_\ast}{R_\ast}-\frac{\left(m-m_\ast\right)   }{ R_\infty^3 }  37.5m^2 \frac{\Delta}{\eta}}.
\label{nb8}
\end{eqnarray}
Assume that the second term is much smaller than $a^2_\infty$; this implies   (beware that $m_\ast/R_\ast$ is not smaller than $a^2_\infty $)
\begin{equation}
37.5  \left(\frac{ m}{R_\infty }\right)^{ 3   } \frac{\Delta}{\eta} \ll a_\infty^2.
\label{nb9}
\end{equation} 
In such a case one obtains
\begin{equation}
k(R)\ge \sqrt{1-2\frac{m_\ast}{R_\ast} }\ge \sqrt{\frac{-\Gamma+2+e}{3\Gamma-2-3e}} .
\label{nb10}
\end{equation}
In the  remaining case, for $R \ge 2.5$, and also when the sonic radius $R_\ast$ exceeds $2.5m$, we get for $R \ge R_\ast$,
\begin{equation}
k(R)\ge \sqrt{0.2} 
\label{nb11}
\end{equation}
---the second term in bound (\ref{boundkkn}). This accomplishes the proof of Lemma 6.
\end{proof}

\begin{remark}
The expression on the rhs of Eq.\ (\ref{i2}) achieves its maximal value (equal to $1$) at $\Gamma =2$ and $e=0$; in this limiting case inequality (\ref{nb10}) becomes just $k(R)\ge 0$. Sonic points might be  close to the apparent horizon---which is the sphere with $k(R)=0$---but do not lie on it, if $a^2_\infty >0$. We can separate further the sonic and apparent horizons, by assuming a bound onto the asymptotic speed of sound, if exponents $\Gamma$ are close to 2. Put, for instance $a^2_\infty <5/13$ for $\Gamma \ge 18/13$. Then    
$4\frac{\Gamma-1-e}{3\Gamma -2-3e}< 0.9$ for all $\Gamma $'s. Thus  
$R_\ast >2m_\ast/0.9$ and the minimal value of $k(R_\ast)$ might be achieved 
for the polytropic index $\Gamma =2$. In such a case we have $k(R_\ast)>  \sqrt{0.1}$.
\end{remark}

\subsection{Main result}

Lemma 6 guarantees that sonic horizons exist outside apparent horizons if $a^2_\infty >0$---thus $k(R)$ is strictly positive, as explained in the preceding remark. This is needed in the forthcoming calculation.

Define $k_\mathrm{inf} \equiv \inf \left\{ \sqrt{\frac{-\Gamma+2+e}{3\Gamma-2-3e}} , \left( 0.2 \right)^{1/2}   \right\}$.

\begin{thm}
Assume that the density $n$ is a nonincreasing function of $R$ and the following conditions hold:
\begin{enumerate}[i)]
\item $\Gamma \in (1, 2]$, $a^2_\ast \le \Gamma -1-e$; the asymptotic data are
 $\varrho_\infty$, $R_\infty$, $m$, and $a_\infty$;
\item $3 \frac{\Delta}{\eta k_\mathrm{inf}^2}\frac{m}{R_\infty} \ll 1$;
\item $37.5\left(\frac{m}{R_\infty }\right)^{3}\frac{ \Delta}{\eta }\ll a_\infty^2$.
\end{enumerate}
Then the sonic point parameters $a^2_\ast$, $U^2_\ast$, and $m_\ast/R_\ast$ in the above model with backreaction are essentially the same as in the test fluid accretion with the same asymptotic data.
\end{thm}
   
\begin{proof}
Conditions ii) and iii) allow us to use estimate (\ref{boundkkn}) of Lemma 6,
\begin{equation}
k \ge  k_\mathrm{inf}.
\label{boundkkb}
\end{equation}
Statement i) of Lemma 6 gives us, in the worst scenario,   $k_\mathrm{inf}=\sqrt{0.2}$.

Using this, employing $p \le \varrho$ and estimate (\ref{boundrho1}), we obtain a chain of inequalities 
\begin{eqnarray}
4\pi\int_{R_\ast}^{R_\infty} (  p + \varrho) \frac{s}{k^2 } ds & \le & \nonumber \\
\frac{4\pi}{k^2_{\inf}} \int_{R_\ast}^{R_\infty} (  p + \varrho) s ds & \le & \nonumber\\
\frac{8\pi}{k^2_{\inf}}   \int_{R_\ast}^{R_\infty}  n_\infty \Delta s ds & \le & \nonumber \\
3 \frac{\Delta }{k^2_{\inf}\eta} \frac{m - m(R_\ast)}{R_\infty} & \ll & 1. \label{e11}
\end{eqnarray}
In order to get the last but one inequality, we exploited the bound $\frac{4\pi}{3} \varrho_\infty \left(R_\infty^3-R_\ast^3\right)\le m-m_\ast$.

Using $p= \frac{a^2(\Gamma -1)}{\Gamma(\Gamma -1 - a^2)}n$, $a^2_\ast \le \Gamma -1-e$, and (\ref{obvious}), one gets
\begin{eqnarray}
c_\ast & = & 2 \pi R^2_\ast \frac{a^2_\ast(\Gamma -1)}{\Gamma(\Gamma -1 -a^2_\ast)}n_\ast  \le \nonumber\\
&&  a^2_\ast \frac{\Delta}{e\eta} \frac{3m}{4\pi  R_\infty} \frac{R_\ast^2}{R_\infty^2} \ll a^2_\ast .
\label{ne11}
\end{eqnarray}
That implies that $c_\ast$ can be put to zero in all relations between characteristics of the sonic point. Inserting estimates (\ref{e11}) and (\ref{ne11}) into the expression for the lapse $N$, we get the estimate
\begin{equation} 
N(R_\ast)\approx \sqrt{1-2\frac{m_\ast}{R_\ast}+U^2_\ast}=\frac{1}{\sqrt{1+3a^2_\ast}}
\label{e2c}
\end{equation}
The above estimate and the Bernoulli-type equation (\ref{n1}) yield now:
\begin{equation}
N_\ast\left( 1-\frac{a^2_\infty}{\Gamma -1}\right)=1-\frac{a^2_\ast}{\Gamma -1}.
\label{estimateb}
\end{equation}
This is a cubic equation for an unknown $a^2_\ast$:  
\begin{eqnarray}
\frac{3a^6_\ast}{(\Gamma-1)^2}+a^4_\ast\frac{-6\Gamma+7}{(\Gamma-1)^2}+a^2_\ast\frac{3\Gamma -5}{\Gamma-1}+ && \nonumber\\
1-\left(1-\frac{a^2_\infty}{\Gamma -1}\right)^2 & =& 0.
\label{nsonic}
\end{eqnarray}
Coefficients of this algebraic equation   depend only on 
the asymptotic speed of sound.  
Its solution has been obtained in \cite{KM2006} for $\Gamma \le 5/3$, with a sign mistake (see also \cite{mandal_ray_das}). The correct formula is given below
\begin{eqnarray}
&&a_\ast^2=\frac{6\Gamma-7}{9  }+ 2\frac{3\Gamma-2}{9} \times \nonumber\\
&& \cos\left[\frac{\pi}{3}+\frac{1}{3}\arccos\left( \frac{G+486(\Gamma-1)a_\infty^2-243a_\infty^4 }{2(3\Gamma-2)^3}\right)\right],
\nonumber\\
\label{slns0}
\end{eqnarray}
where $G\equiv 54\Gamma^3-351\Gamma^2+558\Gamma-259$. Let us point out, that (\ref{slns0}) solves the sonic point equation (\ref{nsonic}) also for $\Gamma \in [5/3, 2]$. 
 
Chaverra et al.\ have found other solutions for $\Gamma \in (5/3,2]$---they are related to a family of homoclinic solutions, with which we do not deal in this paper \cite{CMS}.  
\end{proof}

\section{Numerical analysis}
\label{numerics}

\subsection{On the numerical method}
The following subsections report numerical examples with the polytropic equation of state $p = Kn^\gamma$.

We solve numerically the following set of equations:
\begin{subequations}
\label{ddx}
\begin{eqnarray}
\label{da2dx}
\frac{da^2}{dx} & = & - a^2 \frac{\Gamma - 1 - a^2}{k^2 a^2 - U^2} \nonumber \\
&& \times \left[ \frac{m(x)}{\exp(x)} - 2 U^2 + 4 \pi \exp(2 x) p \right], \\
\frac{d m}{dx} & = & 4 \pi \exp(3 x) \varrho,
\label{dmdx}
\end{eqnarray}
\end{subequations}
where $x = \ln R$ is a new independent variable. Here $n$ is given by Eq.\ (\ref{nb1}), $\varrho$ is computed from Eq.\ (\ref{nb3}), the pressure is given by
\begin{equation}
p = \frac{\Gamma - 1}{\Gamma} \frac{n a^2}{\Gamma - 1 - a^2},
\end{equation}
and $U^2$ is computed as
\begin{equation}
U^2 = \left( \frac{\dot m_\mathrm{B}}{4 \pi R^2 n} \right)^2.
\end{equation}
The curvature $k$ is given by Eq.\ (\ref{p}).

We assume the parameters: $m$, $R_\infty$, $a_\infty$, $\Gamma$. Equations (\ref{ddx}) are integrated starting from $x_\infty = \ln R_\infty$, assuming boundary conditions $a^2(x_\infty) = a^2_\infty$ and $m(x_\infty) = m$. We use a fairly standard Runge-Kutta method of 8-th order \cite{heirer}. The value of $\dot m_\mathrm{B}$ is adjusted (with a suitable bisection method) to yield a transonic solution. Finding a solution passing through a sonic point requires handling of the $0/0$ value appearing on the right-hand side of Eq.\ (\ref{da2dx}) at the sonic point. In practice, we approximate at the sonic point the right-hand side of Eq.\ (\ref{da2dx}) with its value from a preceding Runge-Kutta step. Apart from checking the sonic-point conditions, we also check for an occurrence of the apparent horizon. This yields the black hole mass $m_\mathrm{BH} = m(x_\mathrm{BH})$, where $x_\mathrm{BH}$ is the coordinate corresponding to the apparent horizon.

Equation (\ref{da2dx}) follows from equation (\ref{N/k}).
The calculation is rather lengthy, but elementary. One has to replace the lapse $N$ using Eq.\ (\ref{bernoulligeneral}) and use the relation
\begin{equation}
    \frac{\varrho + p}{n} = \frac{\Gamma - 1}{\Gamma - 1 -a^2}.
\end{equation}
Next, the fact that $R^2 U n = \mathrm{const}$ yields
\begin{equation}
    \frac{dU^2}{dR} = -\frac{4 U^2}{R} - \frac{2 U^2}{n} \frac{dn}{dR}.
\end{equation}
After some algebra, one arrives at
\begin{eqnarray}
\frac{1}{a^2} \frac{d a^2}{dR} & = &  - \frac{\Gamma - 1 - a^2}{(k^2 a^2 - U^2) R} \nonumber \\
&& \times \left[ \frac{m(R)}{R} - 2 U^2 + 4 \pi R^2 p \right].
\end{eqnarray}
It remains to change the independent variable from $R$ to $x$ to get Eq.\ (\ref{da2dx}). Equation (\ref{dmdx}) follows directly from Eq.\ (\ref{masseq}).

\subsection{How good are analytic estimates? The case of $\Gamma <5/3$.}

It seems to be  easier to satisfy  analytical restrictions of part iii) of Theorem 7 than those of part ii). This is illustrated by two diagrams in Figs.\ \ref{fig:ii} and \ref{fig:iii}, in which the asymptotic speed of sound $a^2_\infty =0.01$. In both cases we have the polytropic index $\Gamma $ as the abscissa. The ordinate in Fig.\ \ref{fig:ii} is equal to
\begin{equation}
I_3 (\Gamma) \equiv \left(37.5 \frac{\Delta}{a_\infty^2 \eta }\right)^{1/3};
\end{equation}
Assumption iii) of Theorem 7 demands now that $R_\infty \gg I_3 m$.

\begin{figure}
\centering
\includegraphics[width=1\columnwidth]{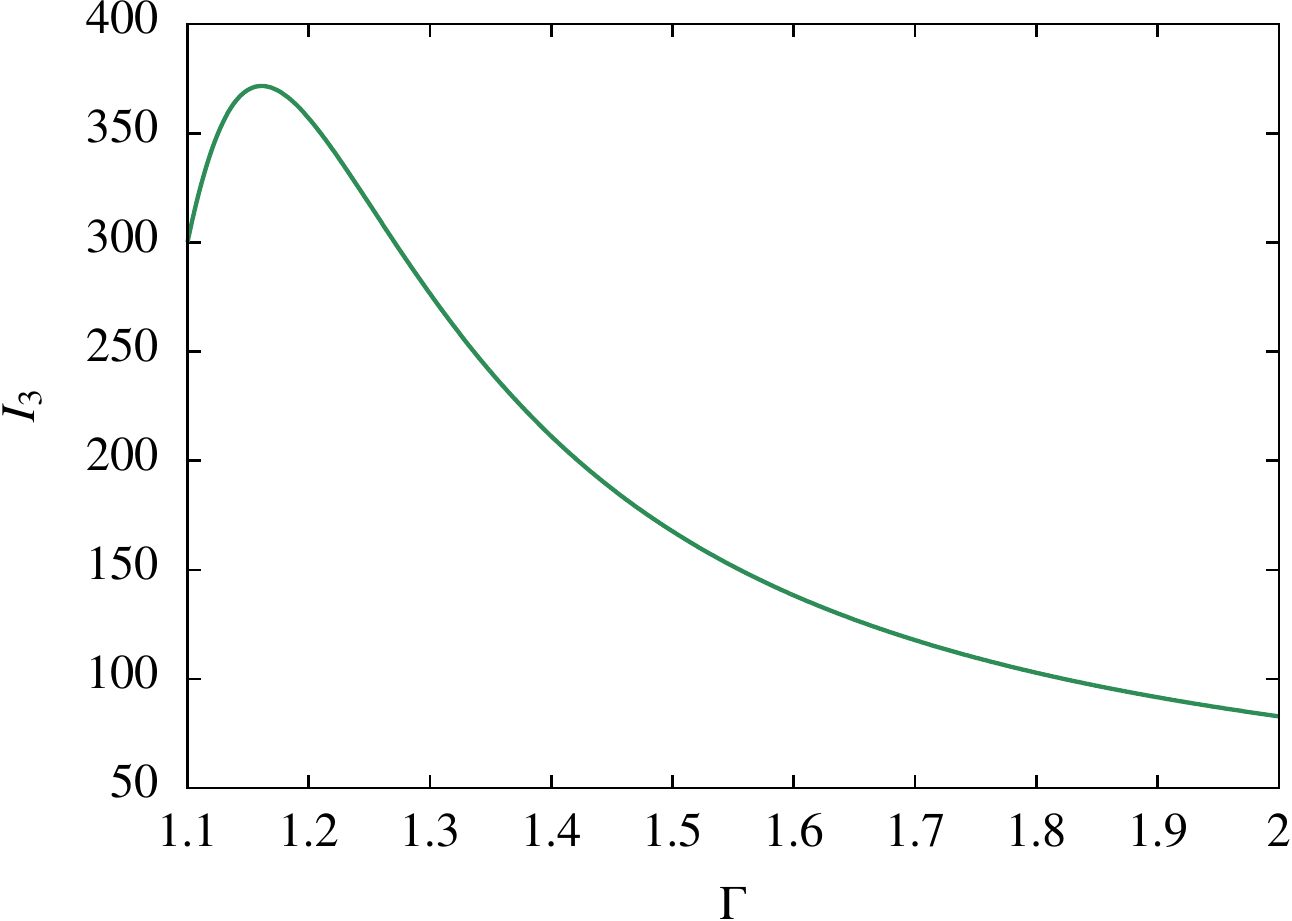}
\caption{The ordinate is $I_3(\Gamma )$, and $\Gamma \in [1.1,2]$ is the abscissa. Here  $a^2_\infty =10^{-2}$.}
\label{fig:ii}
\end{figure}  

The ordinate in Fig.\ \ref{fig:iii} is equal to
\begin{equation}
I_{2}(\Gamma) \equiv 3 \frac{\Delta\left(3\Gamma -2-3e\right)}{\eta (-\Gamma+2+e) };
\end{equation} 
we decided to choose here $k_\mathrm{inf}=\sqrt{\frac{-\Gamma+2+e}{3\Gamma -2-3e}}$. This yields somewhat smaller values of $I_2$ for $\Gamma <1.6$ than the alternative choice $k_\mathrm{inf}=\sqrt{0.2}$, but it is more accurate for $\Gamma \approx 2$, which is the region of interest because of possible astrophysical applications.
Assumption ii) of Theorem 7 demands now that $R_\infty \gg I_2 m$. It is clear from the inspection of the two diagrams that  $I_2$ is much bigger than $I_3$; Thus if $R_\infty \gg I_2m$, then also $R_\infty \gg I_3m$.

\begin{figure}
\centering
\includegraphics[width=1\columnwidth]{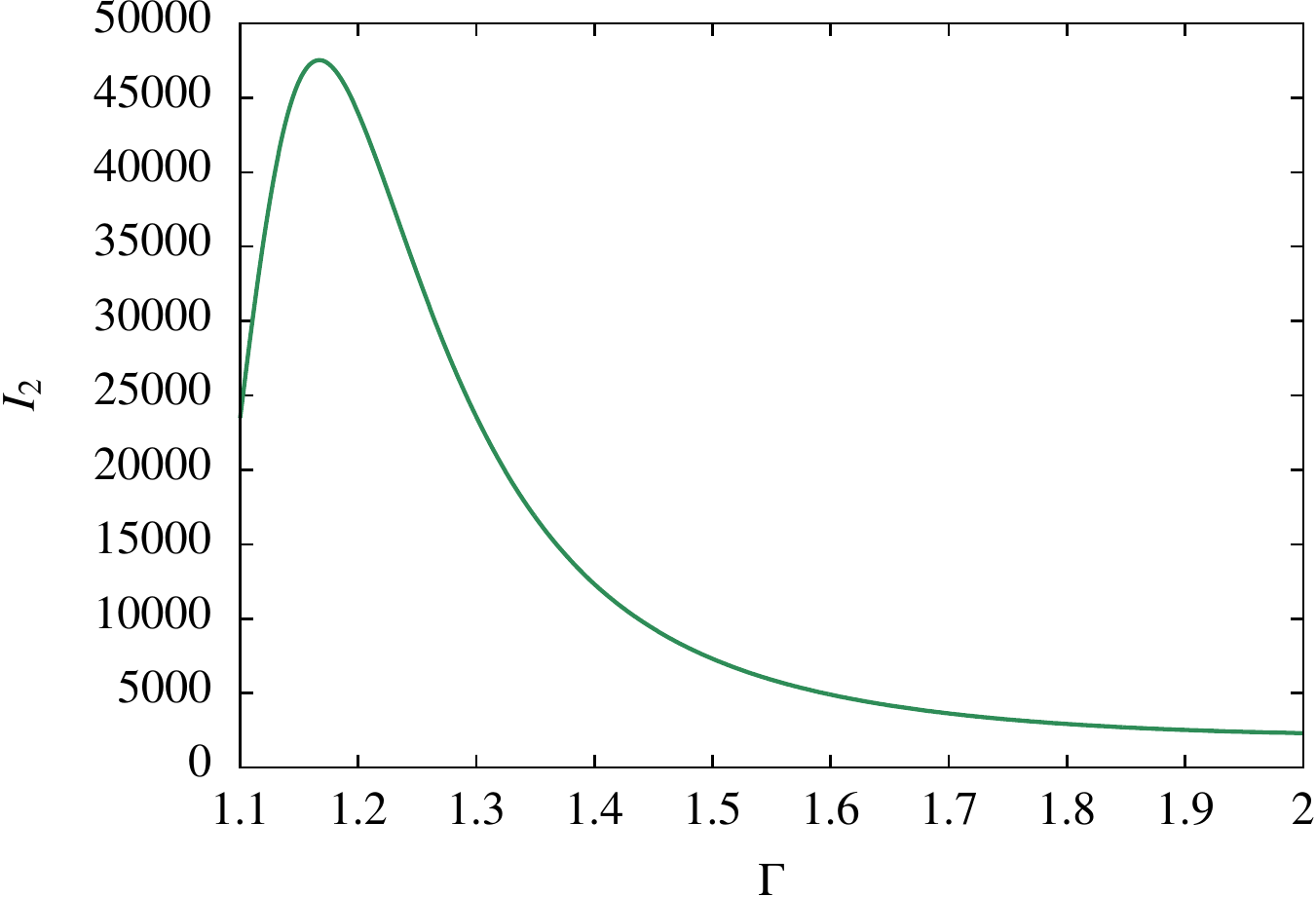}
\caption{The ordinate is $I_2(\Gamma )$, and $\Gamma \in [1.1,2]$ is the abscissa. Here $a^2_\infty =10^{-2}$.}
\label{fig:iii}
\end{figure}

For given boundary data---$R_\infty$, $a^2_\infty$, $\varrho_\infty$, and the total mass $m$---and the two last inequalities being satisfied, Theorem 7 ensures that parameters of sonic points do not depend on $\varrho_\infty$, that is they do not depend on the ratio $m_\mathrm{BH}/m$.

The two forthcoming diagrams (Figs.\ \ref{fig:3} and \ref{fig:4}) illustrate this fact. Some asymptotic data  are the same for the two cases: $R_\infty = 10^6$, $m_\infty = 1$, and $\Gamma = 1.5$. Asymptotic speeds of sound are different. In the case of $a^2_\infty = 4\times 10^{-4}$ we obtain in Fig.\ \ref{fig:3} the numerical speed of sound at the sonic point $a^2_\ast\approx 0.00158$; this is very close to the analytic value $a^2_\ast = 0.00157945$ obtained from (\ref{slns0}). In this case $4100 m \approx I_3 m \ll R_\infty < I_2 m \approx 4.2 \times 10^6m$; thus one of sufficient conditions is not satisfied, but Theorem 7 is still valid---the speed of sound is independent of $m_\mathrm{BH}/m$. In the other example, with $a^2_\infty =  10^{-2}$, we have found (Fig.\ \ref{fig:4}) the numerical speed of sound $a^2_\ast\approx 0.03188$; this is again close to the analytic value $a^2_\ast = 0.0318652$ obtained from (\ref{slns0}). Now $R_\infty \gg I_2 m \approx 7300 m \gg I_3 m \approx 168 m$, and obviously $a^2_\ast$ does not depend on $m_\mathrm{BH}/m$.

Let us remark, that the Newtonian value, known from the Bondi model, is equal to $a^2_\ast = 2 a^2_\infty/(5 - 3\Gamma)$; this gives $a^2_\ast = 0.0016$ for $a^2_\infty = 0.0004$, close to the corresponding numerical result. 

In the next example we have boundary data $R_\infty = 500$, $m = 1$, $a^2_\infty = 10^{-3}$, and $\Gamma = 1.5$. It appears that now $I_2 \approx 6.7\times 10^5$ and $I_3\approx 1640$. Thus both conditions ii) and iii) of Theorem 7 are broken, quite convincingly. Figure \ref{fig:5} shows that the sonic speed of sound depends strongly on the ratio $m_\mathrm{BH}/m$, as it should be expected.

In conclusion: numerical examples---these discussed above, but also a large pool of others---demonstrate, that analytical assumptions of Theorem 7 might be too strong, in particular for small values of the asymptotic speed of sound. Theorem 7 can be true even if some of conditions ii) or iii) are not satisfied. Its predictive power improves with an increase of $a^2_\infty$ and fails quite strongly for very small values of $a^2_\infty $.
 
\begin{figure}
\centering
\includegraphics[width=1\columnwidth]{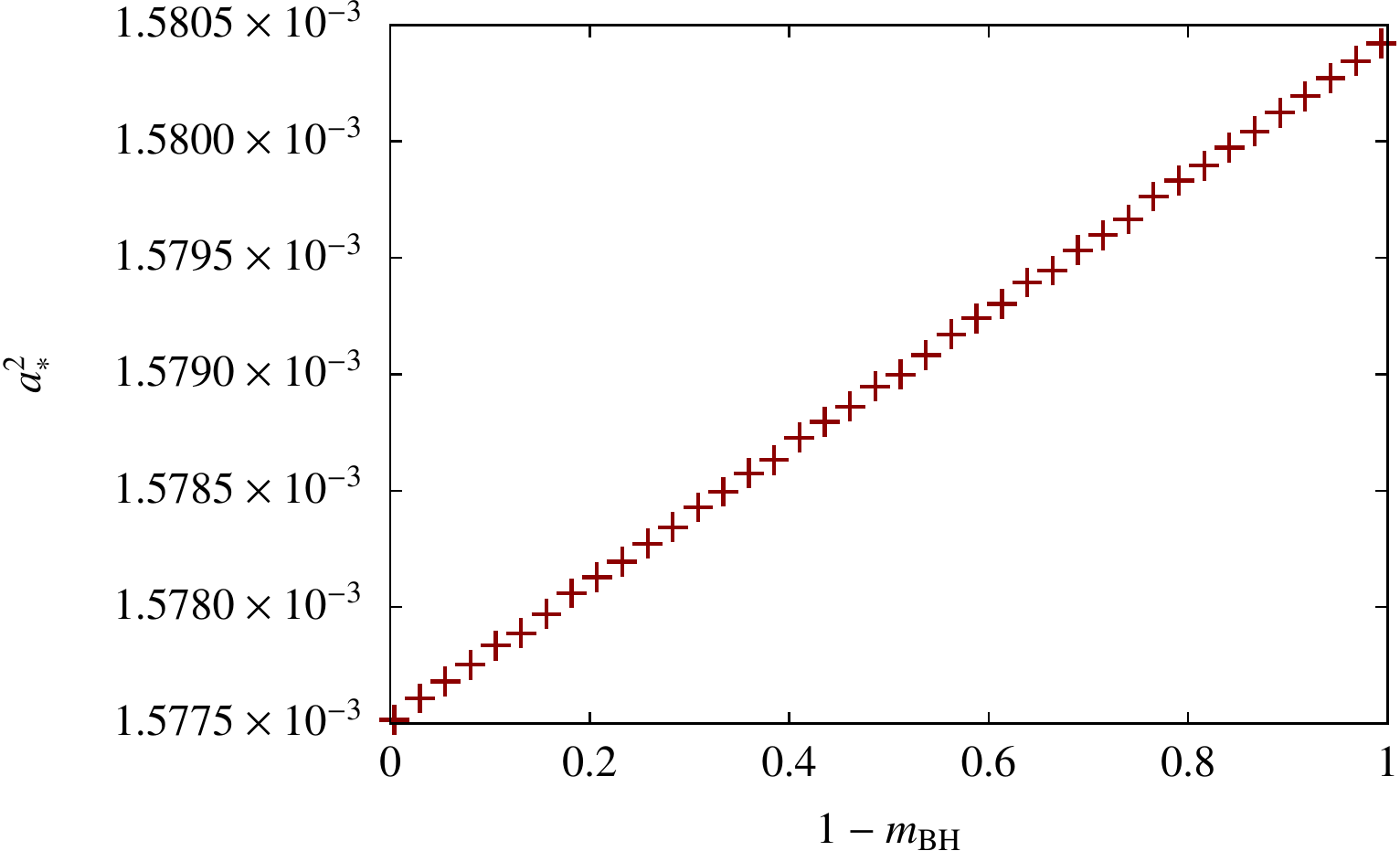}
\caption{$a^2_\ast$ (ordinate)  versus the fluid mass $1 - m_\mathrm{BH}$ (abscissa). The asymptotic parameters are: $m = 1$, $a^2_\infty = 4 \times 10^{-4}$, $\Gamma =1.5$, and $R_\infty =10^6$.}
\label{fig:3}
\end{figure}

\begin{figure}
\centering
\includegraphics[width=1\columnwidth]{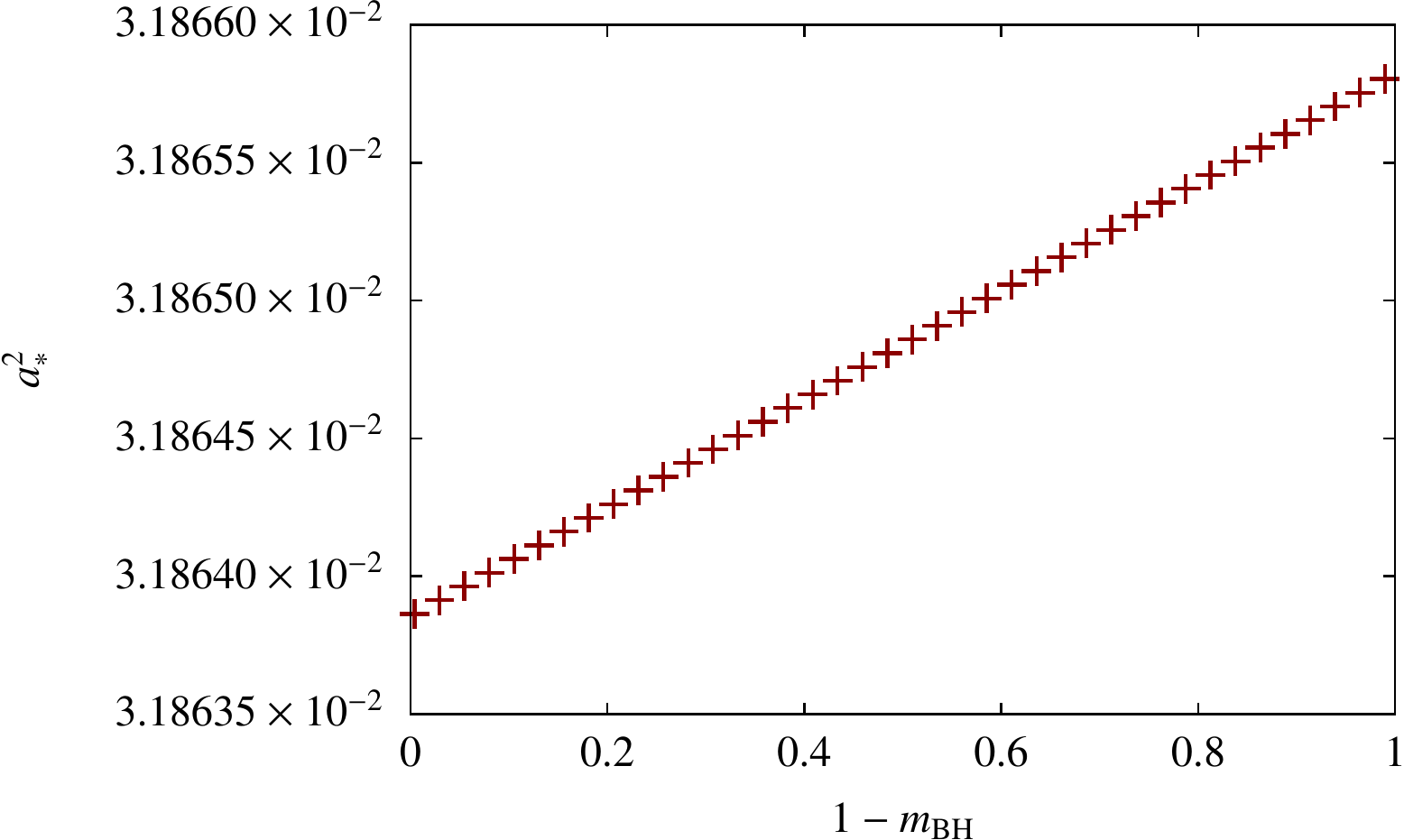}
\caption{$a^2_\ast$ (ordinate) versus the fluid mass $1 - m_\mathrm{BH}$ (abscissa). The asymptotic parameters are: $m=1$, $a^2_\infty = 10^{-2}$, $\Gamma =1.5$, and $R_\infty = 10^6$.}
\label{fig:4}
\end{figure}

\begin{figure}
\centering
\includegraphics[width=1\columnwidth]{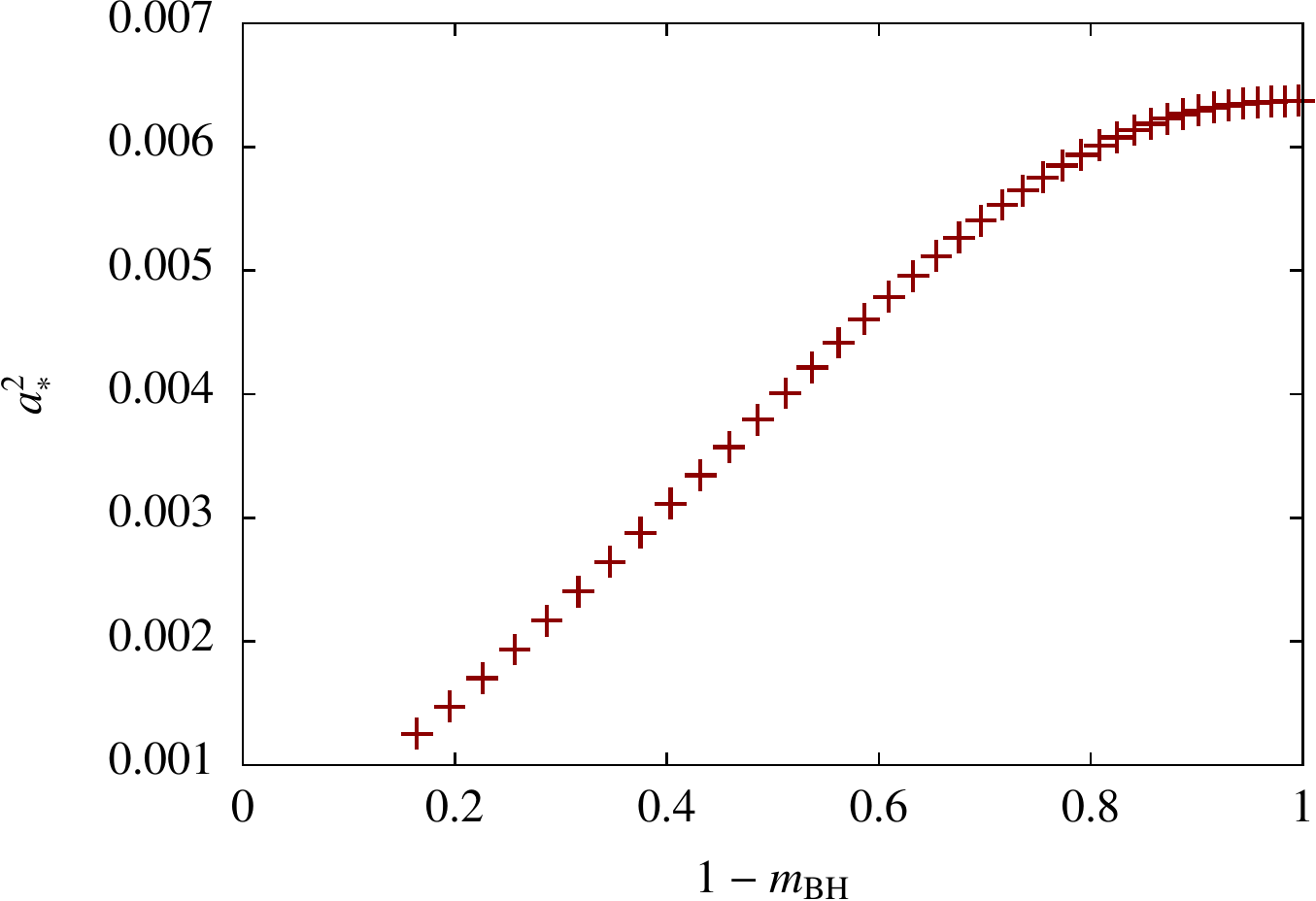}
\caption{$a^2_\ast$ (ordinate) versus the fluid mass $1 - m_\mathrm{BH}$ (abscissa). The asymptotic parameters are: $m = 1$, $a^2_\infty = 10^{-3}$, $\Gamma = 1.5$, $R_\infty = 500$.}
\label{fig:5}
\end{figure}
 
\subsection{Polytropic equation of state with $\Gamma =5/3$}
  
Asymptotic data in the first example (Fig.\ \ref{fig:6}) are $R_\infty = 100$, $a^2_\infty = 10^{-2}$, $m = 1$. We assume the polytropic index $\Gamma = 5/3$. Now  $I_2 m \approx 4000 m$ is much larger than the size $R_\infty = 100$, which in turn is marginally smaller than $I_3 m \approx 124 m$. There is a strong dependence of $a^2_\infty $ on the fluid abundance $1 - m_\mathrm{BH}/m$.

\begin{figure}
\centering
\includegraphics[width=1\columnwidth]{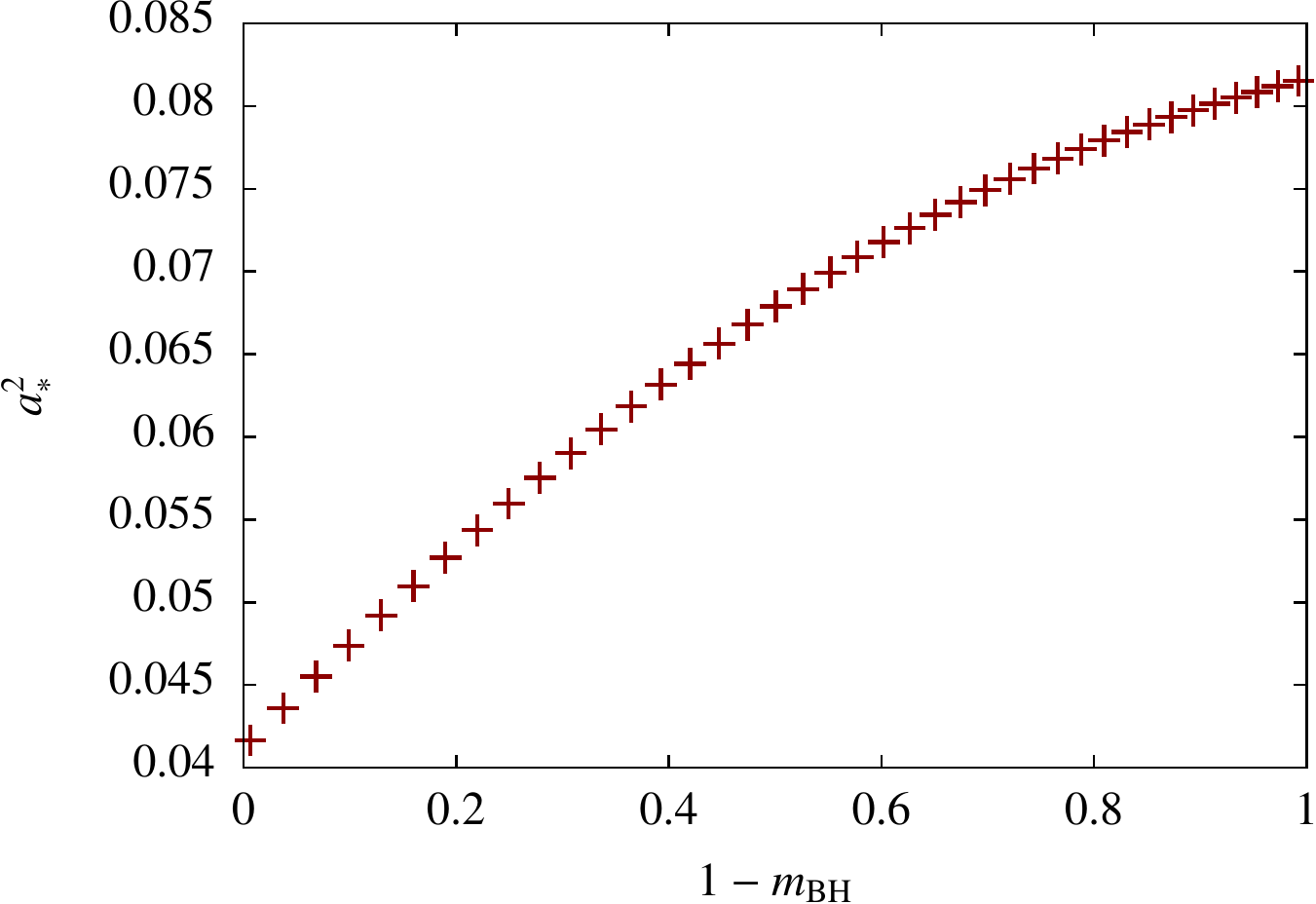}
\caption{$a^2_\ast$ (ordinate) versus the fluid mass $1 - m_\mathrm{BH}$ (abscissa). The asymptotic parameters are: $m=1$, $a^2_\infty = 10^{-2}$, $\Gamma =5/3$, and $R_\infty = 100$.}
\label{fig:6}
\end{figure}

In the next example we consider the same values of the asymptotic speed of sound ($a^2_\infty = 10^{-2}$), but enlarge the size to $R_\infty = 10^4$. Now $  R_\infty >I_2 m \gg I_3m$. It appears (see Fig.\ \ref{fig:7}) that $a_\ast^2$ depends rather weakly on $1-m_\mathrm{BH}/m$.  
 
\begin{figure}
\centering
\includegraphics[width=1\columnwidth]{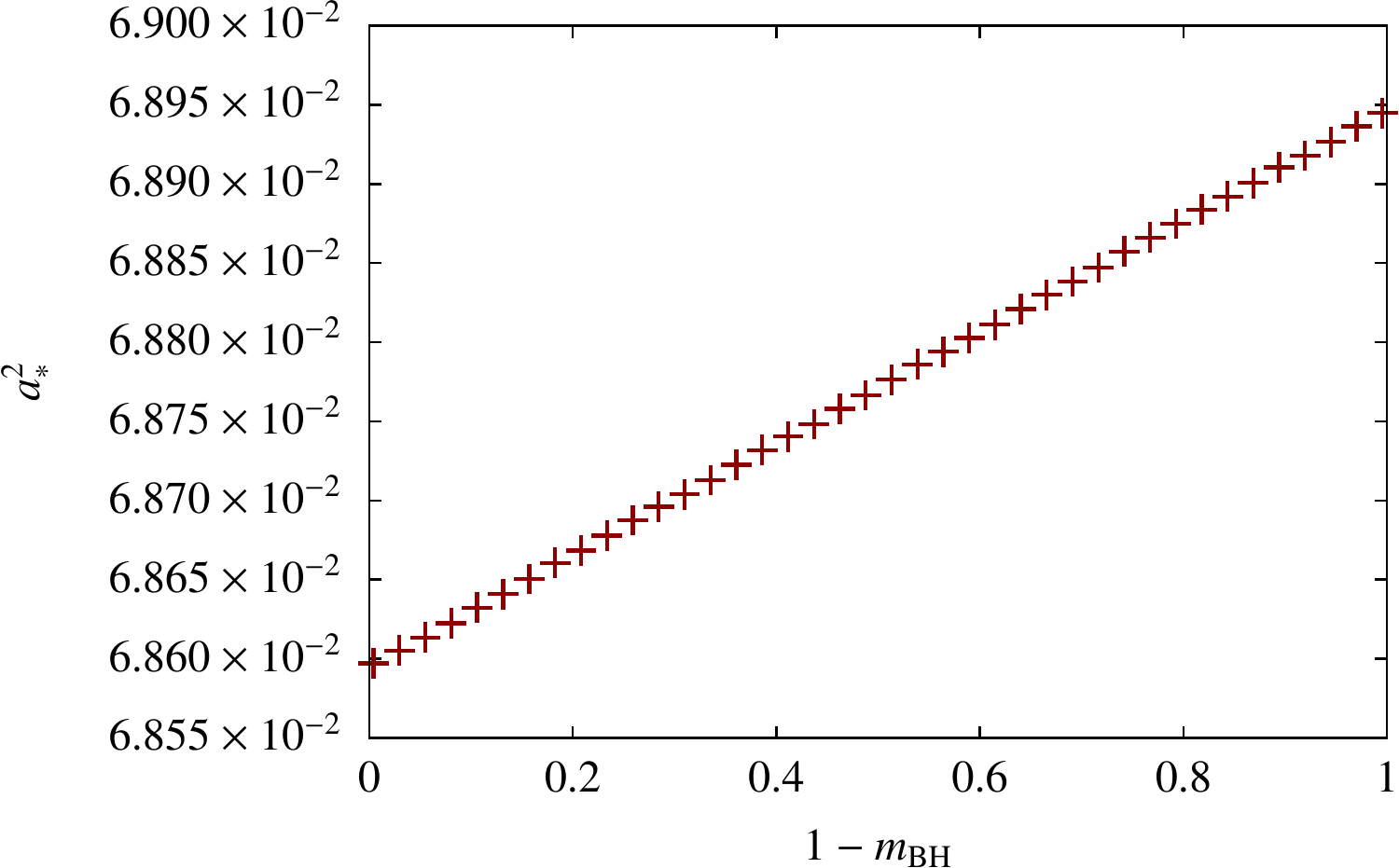}
\caption{$a^2_\ast$ (ordinate) versus the fluid mass $1 - m_\mathrm{BH}$ (abscissa). The asymptotic parameter are: $m=1$, $a^2_\infty = 10^{-2}$, $\Gamma =5/3$, and $R_\infty = 10^4$.}
\label{fig:7}
\end{figure}

The increase of the size $R_\infty$ by another factor of 100 yields a picture with essentially constant sonic speed $a^2_\ast \approx 0.068 $ (see Fig.\ \ref{fig:8}). For comparison, the analytic value obtained from (\ref{slns0}) is 
$a^2_\ast = 0.068827$. Notice that now $R_\infty $ is now much larger than $I_2 m$.

\begin{figure}
\centering
\includegraphics[width=1\columnwidth]{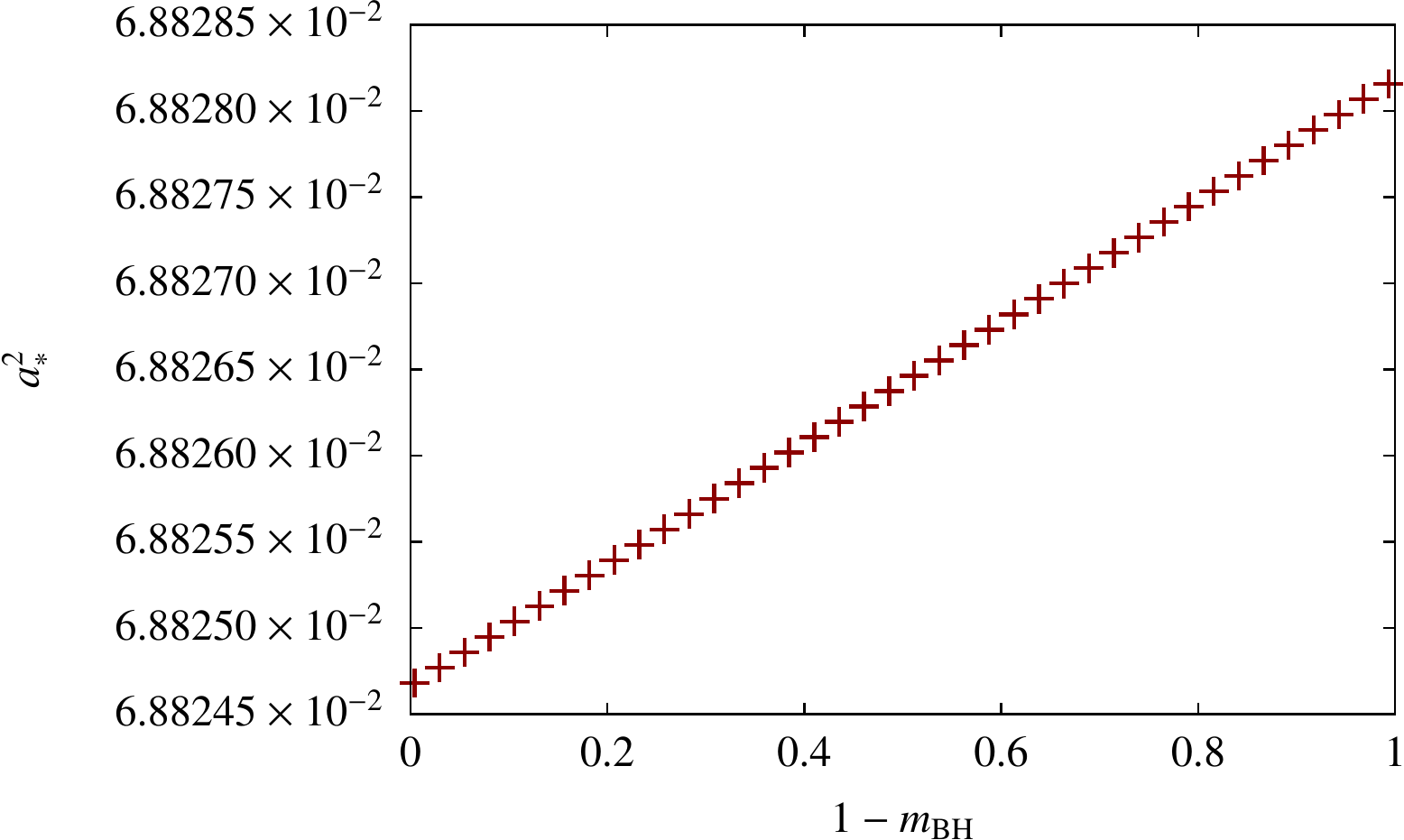}
\caption{$a^2_\ast$ (ordinate) versus the fluid mass $1 - m_\mathrm{BH}$ (abscissa). The asymptotic parameters are: $m = 1$, $a^2_\infty = 10^{-2}$, $\Gamma = 5/3$, and $R_\infty = 10^6$.}
\label{fig:8}
\end{figure}

We shall also investigate the influence of the asymptotic speed of sound. Thus we take as before $R_\infty = 10^4$ but 10 times smaller asymptotic speed of sound, $a^2_\infty = 10^{-4}$. Now $I_2 m \approx 3.8 \times 10^6 m$ is much larger than the size $R_\infty = 10000$, which in turn is larger than $I_3 m \approx 5700 m$. There is a clear dependence (see Fig.\ \ref{fig:9}) of $a^2_\ast $ on the fluid mass fraction within the system.

\begin{figure}
\centering
\includegraphics[width=1\columnwidth]{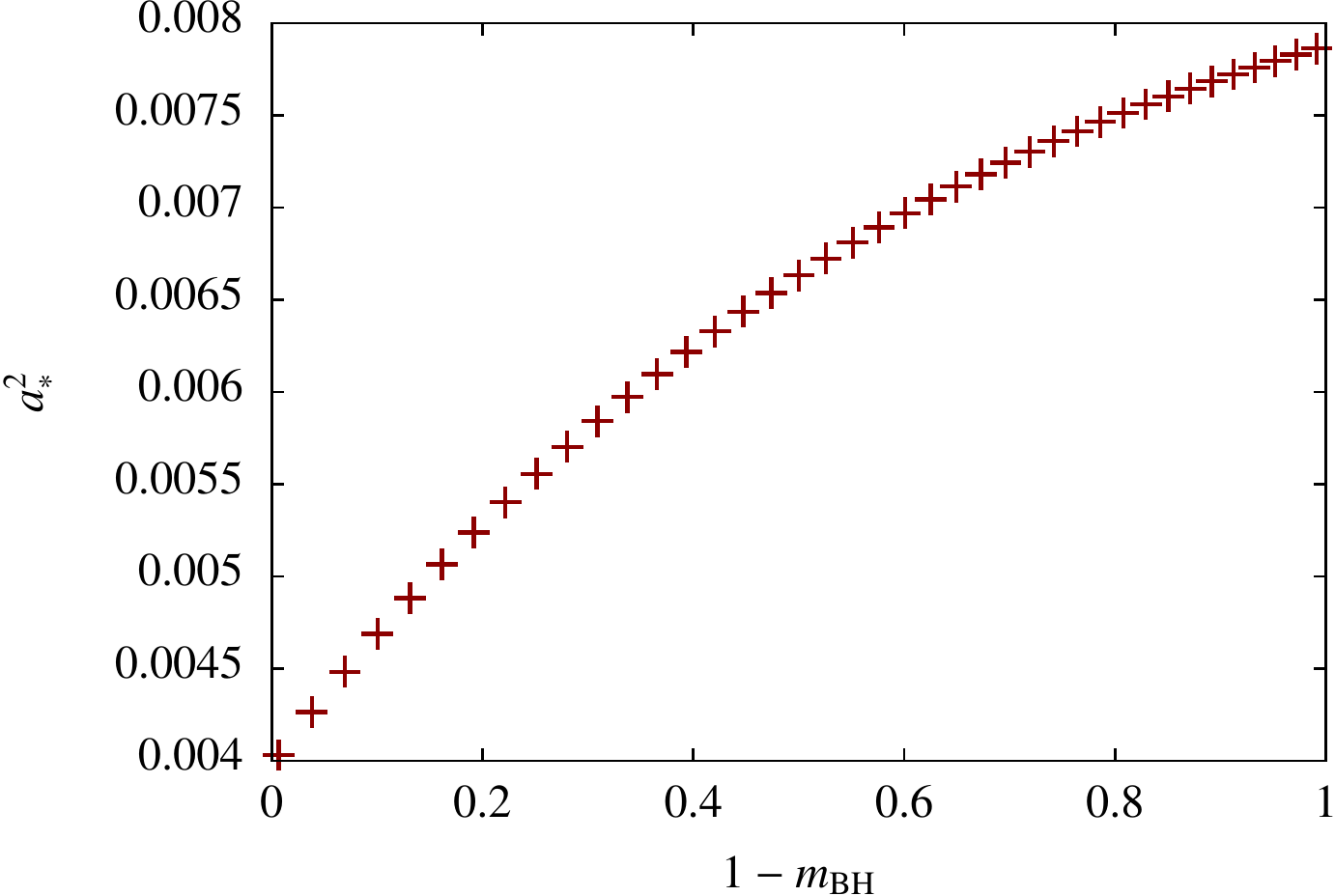}
\caption{$a^2_\ast$ (ordinate) versus the fluid mass $1 - m_\mathrm{BH}$ (abscissa). The asymptotic parameters are: $m = 1$, $a^2_\infty = 10^{-4}$, $\Gamma=5/3$, and $R_\infty = 10^4$.}
\label{fig:9}
\end{figure}

The increase of the size of the system to $R_\infty = 10^6$, with other parameters the same as in the previous diagram (Fig.\ \ref{fig:9}), yields a different picture, with essentially constant $a^2_\ast \approx 0.00668$ (see Fig.\ \ref{fig:10}). The exact analytic value, found from formula (\ref{slns0}), is $a^2_\ast=0.00668882$.

\begin{figure}
\centering
\includegraphics[width=1\columnwidth]{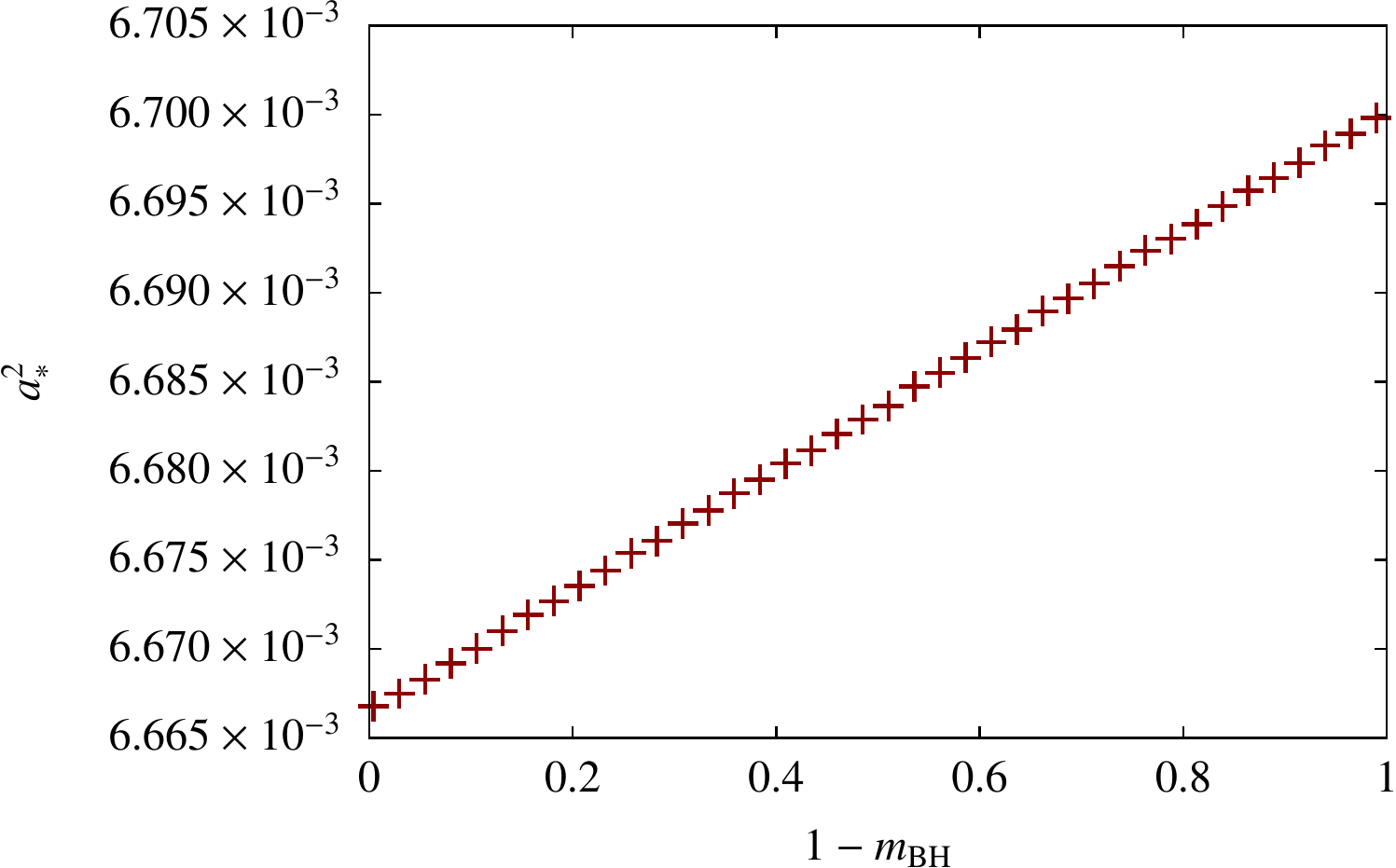}
\caption{$a^2_\ast$ (ordinate) versus the fluid mass $1 - m_\mathrm{BH}$ (abscissa). The asymptotic parameters are: $m=1$, $a^2_\infty = 10^{-4}$, $\Gamma =5/3$, and $R_\infty = 10^6$.}
\label{fig:10}
\end{figure}
 
Asymptotic data in the example depictured in Fig.\ \ref{fig:11} are the same as in Fig.\ \ref{fig:10}, with the exception of  $a^2_\infty = 10^{-6}$. Now $I_2 m \approx 3.8 \times 10^9 m$ is many orders larger than the size $R_\infty = 10^6$, which in turn is a few times larger than $I_3 m \approx 260000 m$. There is pronounced variability of $a^2_\ast$ as a function of $1 - m_\mathrm{BH}$.

\begin{figure}
\centering
\includegraphics[width=1\columnwidth]{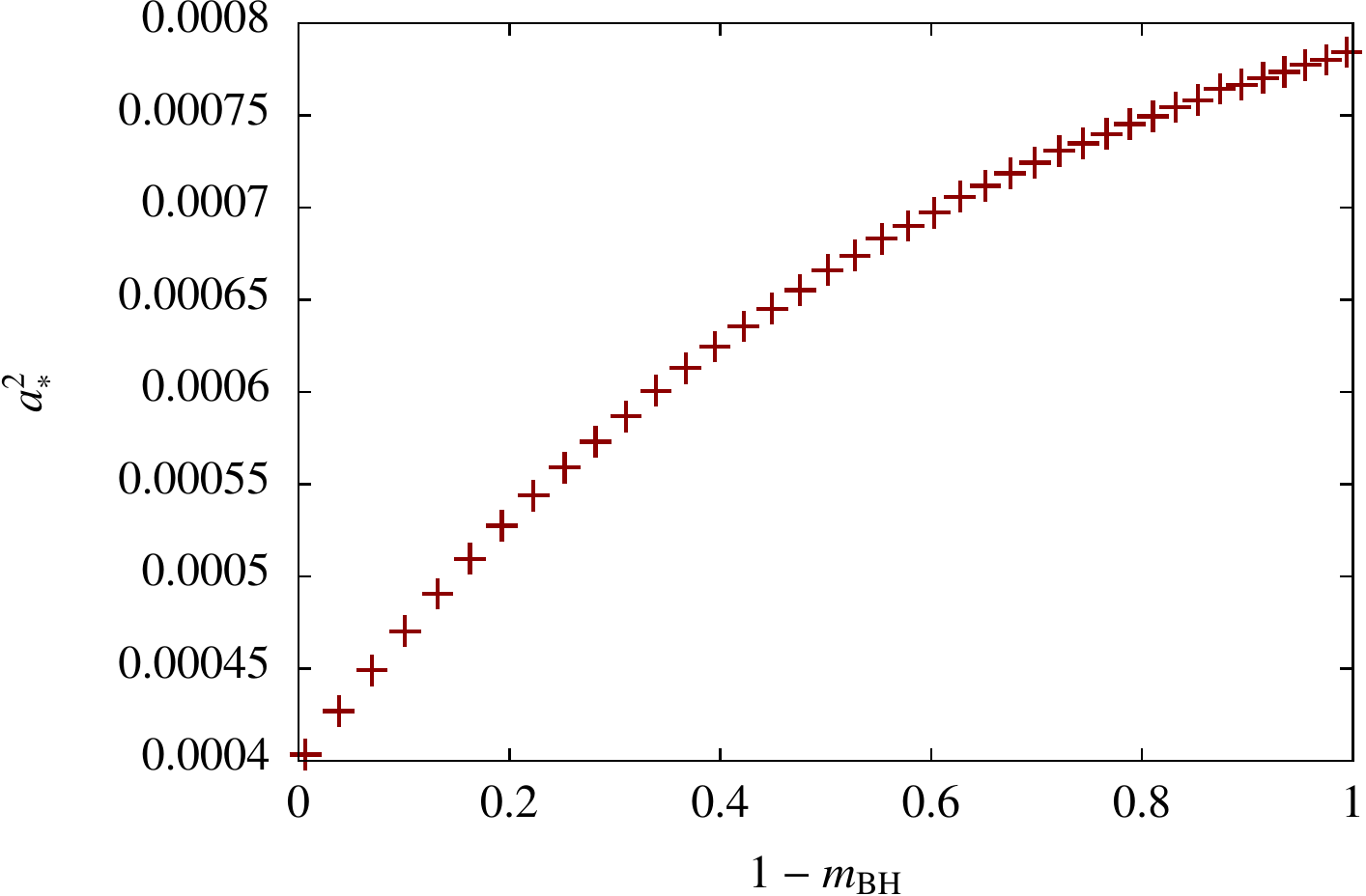}
\caption{$a^2_\ast$ (ordinate) versus the fluid mass $1 - m_\mathrm{BH}$ (abscissa). The asymptotic parameters are: $m = 1$, $a^2_\infty = 10^{-6}$, $\Gamma = 5/3$, and $R_\infty = 10^6$.}
\label{fig:11}
\end{figure}
 
\subsection{Solutions with $\Gamma =2$}

The class of accretion solutions, for polytropic gases with a polytropic index $\Gamma >5/3$, is markedly different from those considered earlier, with $\Gamma \le 5/3$. In contrast to the former case, there might appear homoclinic solutions \cite{CMS}, for values of asymptotic speeds $a^2_\infty$ that might be only somewhat larger than $m/R_\infty $. An exemplary solution is shown in Fig.\ \ref{fig:11b}. We need to exclude them, which complicates numerical analysis. In all examples of this subsection the asymptotic mass is $m = 1$, $R_\infty =100$, and the polytropic index $\Gamma = 2$.

\begin{figure}
\centering
\includegraphics[width=1\columnwidth]{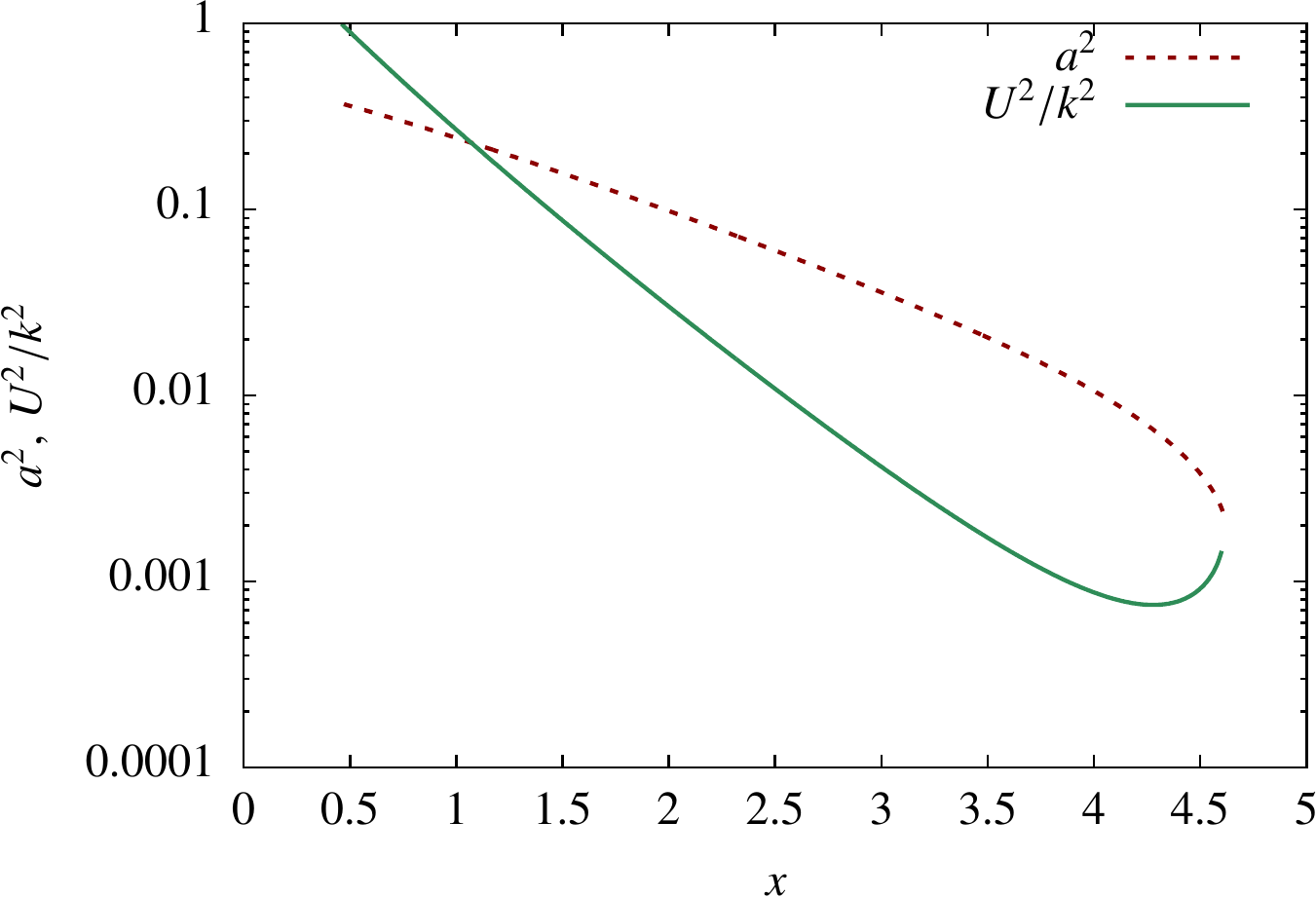}
\caption{An example of a homoclinic-type solution. The graph shows the speed of sound and the velocity in function of $x=\ln R$. The asymptotic parameters are: $R_\infty = 100$, $a^2_\infty = 2.38 \times 10^{-3}$, $m = 1$, $\Gamma = 2$, and $n_\infty = 1.6 \times 10^{-8}$. For this solution $m_\mathrm{BH} = 0.79 \, m$.}
\label{fig:11b}
\end{figure}

In the first two diagrams (Figs.\ \ref{fig:12} and \ref{fig:13}) we present solutions corresponding to $a^2_\infty = 10^{-2}$. Figure \ref{fig:12} shows that the speed of sound of the sonic point depends on $1-m_\mathrm{BH}/m$, albeit the dependence is moderate---the relative change is smaller than 20\%. Figure \ref{fig:13} demonstrates that the mass accretion rate is of the character found in \cite{PRD2006}; $\dot m \propto m^2_\ast (m-m_\ast)$. This is surprising, because now $a^2_\ast $ depends on $1-m_\mathrm{BH}/m$, in contrast to cases discussed in \cite{PRD2006}.

\begin{figure}
\centering
\includegraphics[width=1\columnwidth]{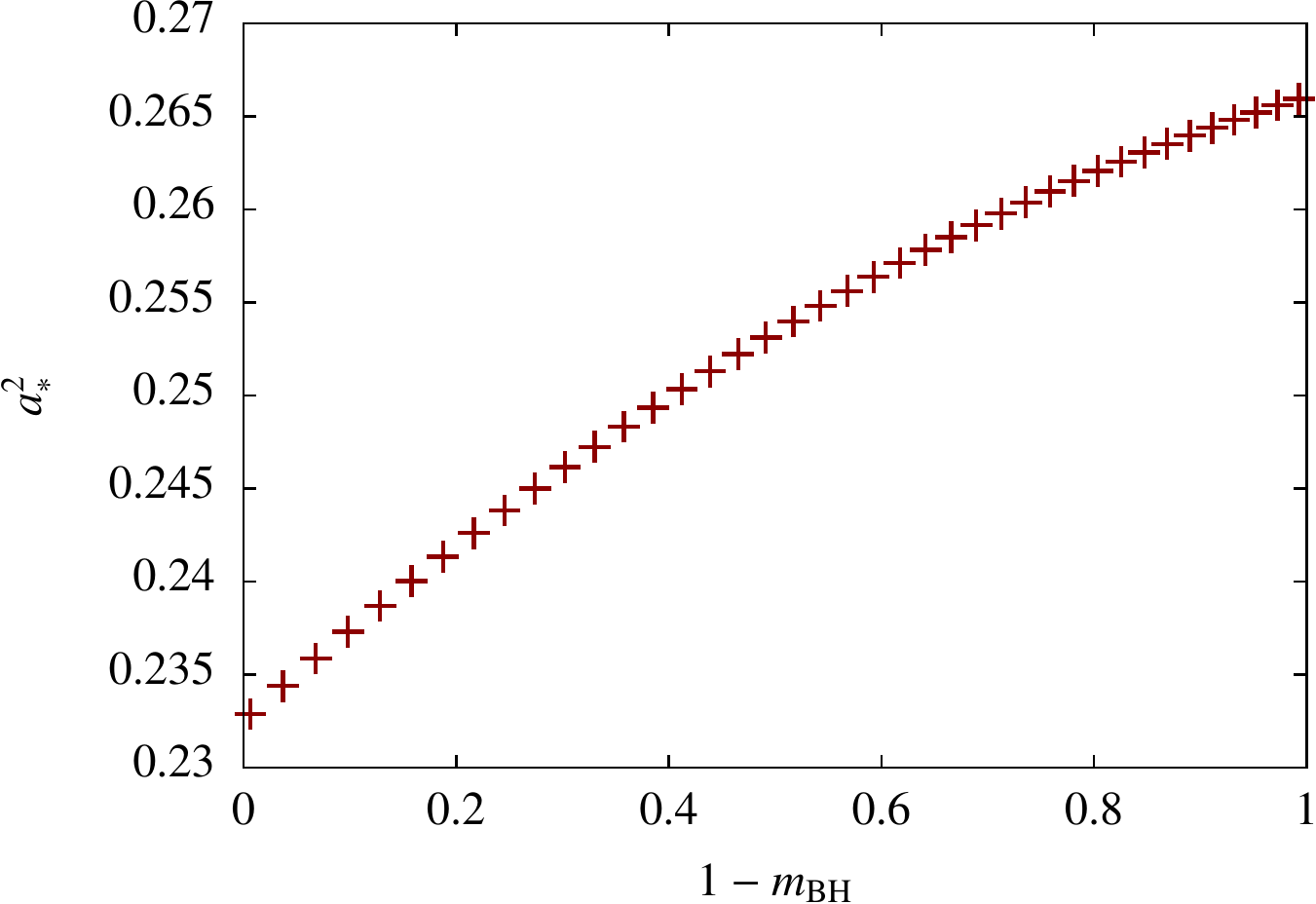}
\caption{$a^2_\ast$ (ordinate) versus the fluid mass $1 - m_\mathrm{BH}$ (abscissa). The asymptotic parameters are: $m = 1$, $a^2_\infty = 10^{-2}$, $\Gamma = 2$, and $R_\infty = 100$.}
\label{fig:12}     
\end{figure}

\begin{figure}
\centering
\includegraphics[width=1\columnwidth]{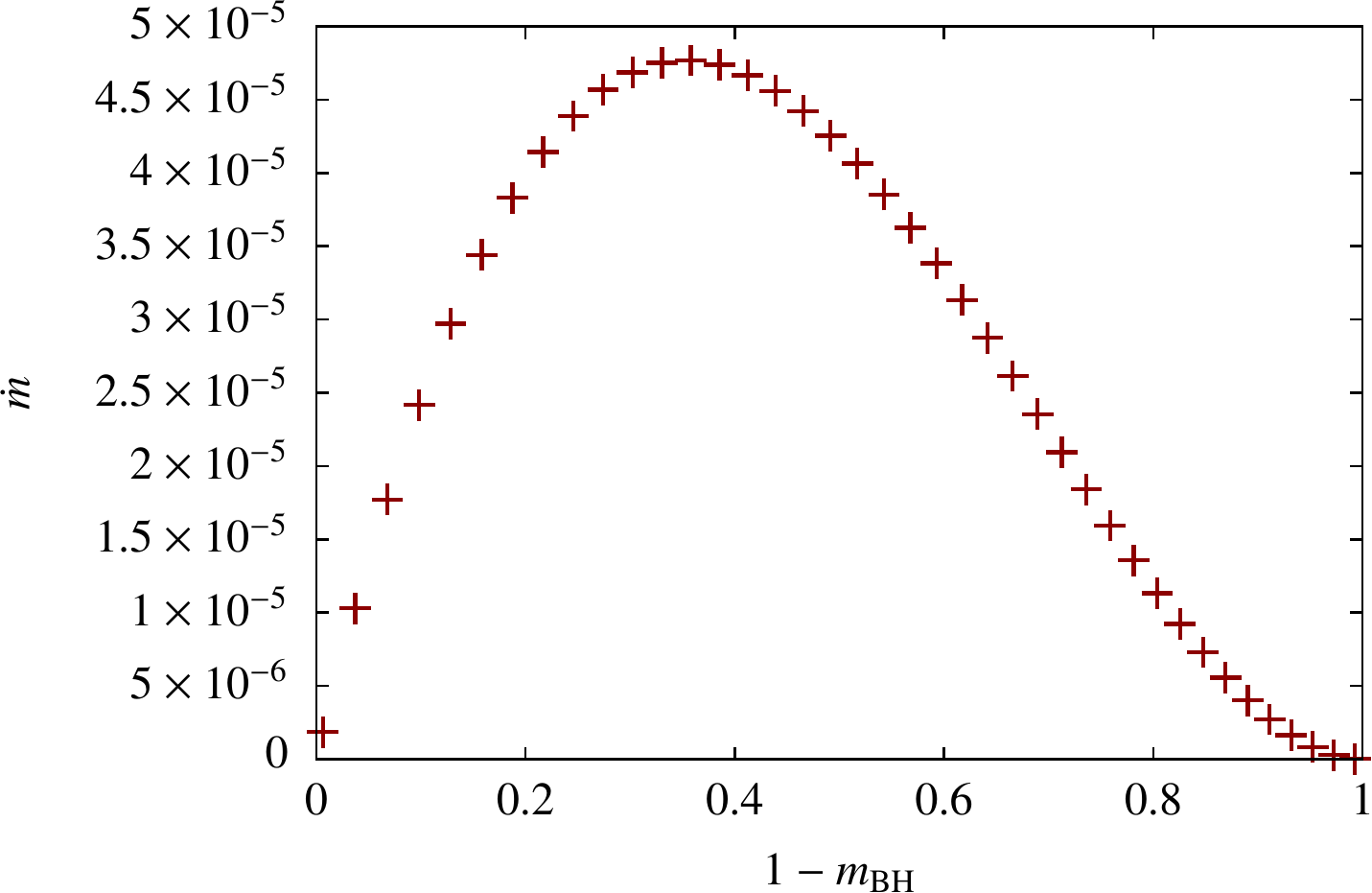}
\caption{The mass accretion rate $\dot m$ (ordinate) versus the fluid mass $1 - m_\mathrm{BH}$ (abscissa). The asymptotic parameters are: $m = 1$, $a^2_\infty = 10^{-2}$, $\Gamma = 2$, and $R_\infty = 100$.}
\label{fig:13}   
\end{figure}

The last two diagrams (Figs.\ \ref{fig:14} and \ref{fig:15}) show solutions corresponding to $a^2_\infty = 10^{-1}$. Figure \ref{fig:14} reveals a small dependence of $a^2_\ast$  on $1-m_\mathrm{BH}/m$---the relative change is smaller than 5\%. The speed of sound $a^2_\ast\approx 0.385$ is close to the analytic value predicted by Eq.\ (\ref{slns0}), according to which $a^2_\ast =0.388336$. Figure \ref{fig:15} again demonstrates that the mass accretion rate is of the character found in \cite{PRD2006}; $\dot m \propto m^2_\ast (m - m_\ast)$.   

\begin{figure}
\centering
\includegraphics[width=1\columnwidth]{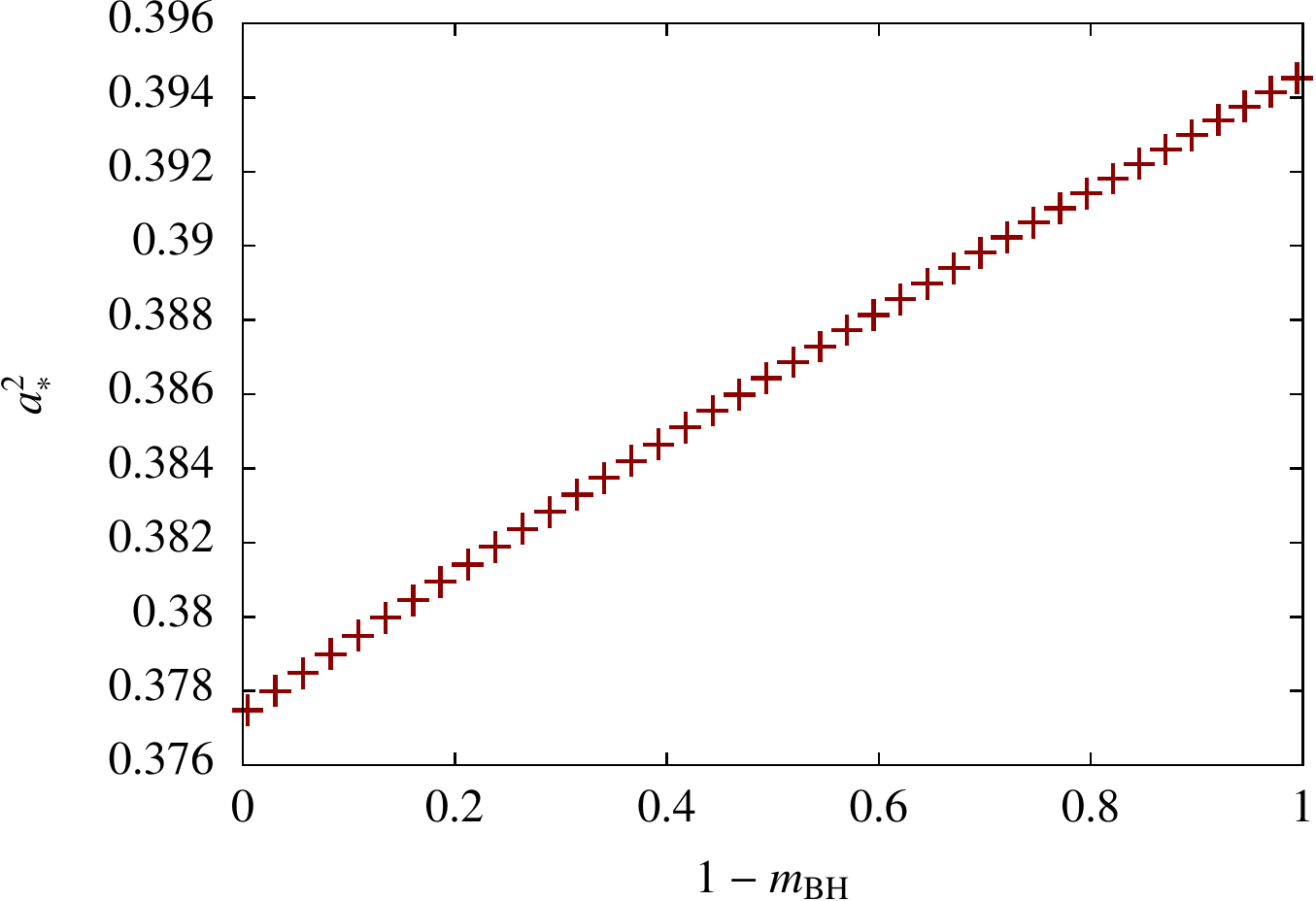}
\caption{$a^2_\ast$ (ordinate) versus the fluid mass $1 - m_\mathrm{BH}$ (abscissa). The asymptotic parameters are: $m = 1$, $a^2_\infty = 10^{-1}$, $\Gamma = 2$, and $R_\infty = 100$.}
\label{fig:14}   
\end{figure}

\begin{figure}
    \centering
\includegraphics[width=1\columnwidth]{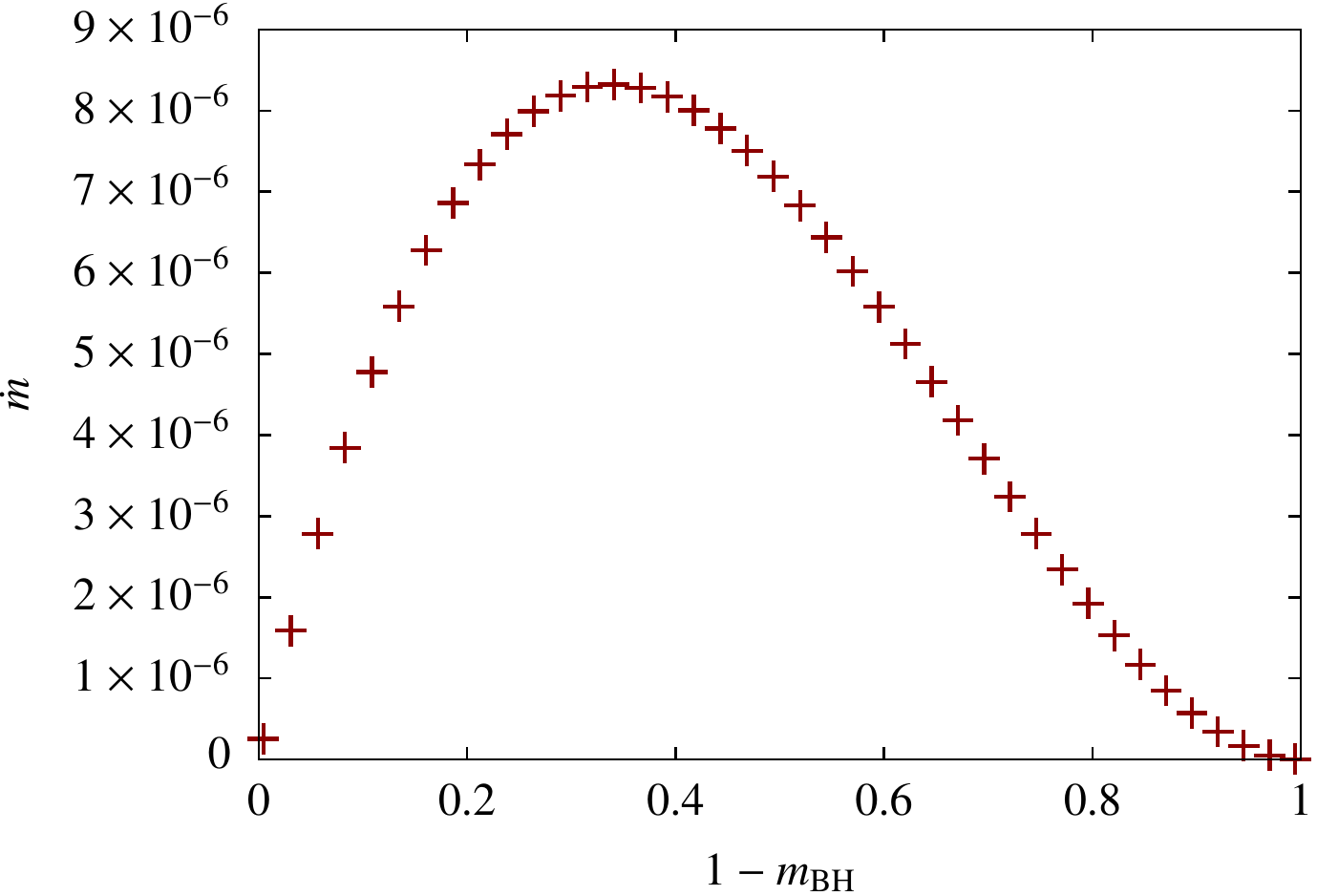}
\caption{The mass accretion rate $\dot m$ (ordinate) versus the fluid mass $1 - m_\mathrm{BH}$ (abscissa). The asymptotic parameters are: $m = 1$, $a^2_\infty = 10^{-1}$, $\Gamma = 2$, and $R_\infty = 100$.}
\label{fig:15}   
\end{figure}

\section{Astrophysical application: lifetime of black holes within polytropic stars}

Richards, Baumgarte, and Shapiro have an interesting idea of estimating the lifetime of hypothetical primordial black holes within neutron stars \cite{Baumgarte_Shapiro, BS}. They assume that matter steadily falls onto a black hole, that the accretion is critical, and that intensive characteristics of sonic points---the speed of sound $a_\ast$, the infall velocity $|U_\ast|/k_\ast $, and the ratio $m_\ast/R_\ast$---do not change during the whole process. The mass accretion rate can be expressed by formula (\ref{dot1m}), i.e.,
\begin{equation}
\dot m = -(\partial_t - (\partial_t R)\partial_R) m_\mathrm{ext}(R) = - 4 \pi N_\ast R^2_\ast U_\ast  (\varrho_\ast  + p_\ast ).
\label{dot1mf}
\end{equation}
The lapse $N_\ast$ can be obtained from Eq.\ (\ref{bernoulli}) (for a $p$--$\varrho$ equation of state) or from Eq.\ (\ref{nnB}) (for $p$--$n$ fluids). We know from our former discussion, that there exist mathematical models that possess required properties. They are characterized by 
assumptions of our Theorems 3, 4, and 7; their fulfillment guarantees that characteristics of sonic points are the same, as for a related process with test fluids. In addition, one should take into account the fact, that the asymptotic mass of the system is constant during accretion. This leads to the   dependence  $\dot m \propto m^2_\ast (m - m_\ast)$, 
which is clearly seen in Figs.\ \ref{fig:13}, \ref{fig:15}.

\subsection{Lifetime of black holes---the idealized $p$--$\varrho$ case}

We shall describe how one can estimate the life of a black hole in an idealized polytropic star, first discussing the $p$--$\varrho$ case. In the test fluid approximation the asymptotic density $\varrho_\infty$ is constant during the accretion process. In the present description, that includes backreaction, it varies. Numerical calculations suggest that $m_\ast (t)\approx m -\gamma \varrho_\infty(t)$, where $\gamma $ is a constant  \cite{PRD2006}. The whole dependence on $\varrho_\infty$ is contained in the factor $m_\ast^2 \varrho_\infty = m^2_\ast(m - m_\ast)/\gamma $ (see a discussion in \cite{PRD2006}, for $\Gamma <5/3$). In such a case we have
\begin{equation}
\dot m = \pi m_\ast^2  \frac{m-m_\ast} {\gamma }   \left( \frac{a_\ast^2}{a_\infty^2} 
\right)^\frac{(5 - 3 \Gamma)}{2(\Gamma - 1)} \left(1 + \frac{a^2_\ast}{\Gamma } \right) \frac{1 + 3 a_\ast^2}{a^3_\infty}.
\label{dotmm}
\end{equation}
From that we get the  time $T(m_\mathrm{in},M)$ (the proper time of a stationary asymptotic observer)
\begin{eqnarray}
\label{dotT}
\lefteqn{T(m_\mathrm{in},M) = } \nonumber \\
&& \int_{m_{\mathrm{in}}}^M \frac{dm_\ast}{ \pi m_\ast^2  \frac{m-m_\ast} {\gamma }   \left( \frac{a_\ast^2}{a_\infty^2} 
\right)^\frac{(5 - 3 \Gamma)}{2(\Gamma - 1)} \left(1 + \frac{a^2_\ast}{\Gamma } \right) \frac{1 + 3 a_\ast^2}{a^3_\infty}}; \nonumber \\
\end{eqnarray}
the duration of time needed for the increase of mass within a sphere of the areal radius $R_\ast$ from initial $m_\ast$ to final $M-m_\mathrm{in}+m_\ast$. The mass of the central black hole would increase from initial $m_\mathrm{in}$ to $M $ during the time $T(m_\mathrm{in},M)$. The lifetime $T_\mathrm{L}$  of a black hole can be defined as the time $T$ that is needed for the transfer of the whole fluid  mass  $m-m_\mathrm{in} $ into the black hole.
The integral in Eq.\ (\ref{dotT}) can be explicitly evaluated, in our idealized case, with the result
\begin{equation}
T(m_\mathrm{in},M) =  \gamma \frac{\ln \left[ \frac{M\left( m-m_\mathrm{in}\right)}{m_\mathrm{in}\left(m-M\right)} \right]+\frac{m}{m_\mathrm{in}}-\frac{m}{M}}{ \pi m^2 \left( \frac{a_\ast^2}{a_\infty^2} 
\right)^\frac{(5 - 3 \Gamma)}{2(\Gamma - 1)} \left(1 + \frac{a^2_\ast}{\Gamma } \right) \frac{1 + 3 a_\ast^2}{a^3_\infty}}.
\label{dotT2}
\end{equation}
It is clear, that in the limit $M\rightarrow m$ we get infinite time $T(m_\mathrm{in},m)$; the black hole persists for an infinite duration of time within our idealized polytropic star. On the other hand, if we fix a reasonable cutoff---for instance that the mass ratio $M $ of the final black hole  is equal to  a concrete fraction  of the total mass $m$ (for instance $0.999 m$)---then we may get a finite result, with the dominant contribution coming from the term proportional to $m/m_\mathrm{in}$.

\subsection{Lifetime of black holes---the idealized $p$--$n$ case}

We shall first calculate the mass accretion rate given by Eq.\ (\ref{4.18}) for the $p$--$n$ equation of state and with no backreaction, that is assuming $c_\ast =0$. As in the former $p$--$\varrho$ case we write
\begin{equation}
4R_\ast^2 = m^2_\ast \left(\frac{2R_\ast}{m_\ast}\right)^2=m^2_\ast\left(\frac{1+3a_\ast^2}{a^2_\ast}\right)^2
\end{equation}
and
\begin{equation}
U_\ast^2 =\frac{m_\ast}{2R_\ast}=\frac{a^2_\ast}{1+3a^2_\ast}.
\end{equation}
The rest mass density $\varrho$ has to be expressed in terms of the baryonic mass density, according to Eq.\ (\ref{nb3a}). A direct calculation yields following result
\begin{eqnarray}
\dot m & = & \pi m^2_\ast n_\infty  \left(\frac{a_\ast}{a_\infty} \right)^{\frac{5-3\Gamma}{\Gamma-1}}\left(\frac{\Gamma-1-a^2_\infty}{\Gamma-1-a^2_\ast} \right)^{\frac{1}{\Gamma-1}} \nonumber \\
& & \times \frac{(\Gamma - 1)(1+3 a^2_\ast)}{a_\infty^3 (\Gamma - 1 - a^2_\ast)}.
\label{mdotnnowe}
\end{eqnarray}

One can estimate the lifetime of a black hole within an idealized polytropic $p$--$n$ star, similarly as in the $p$--$\varrho$ case. In the test fluid approximation the asymptotic density $\varrho_\infty$ is constant during the accretion process. In the present description, that includes backreaction, it varies. Numerical calculations suggest that $m_\ast (t)\approx m -\gamma n_\infty(t)$, where $\gamma $ is a constant. The whole dependence on $n_\infty$ is contained in the factor $m_\ast^2 n_\infty = m^2_\ast(m - m_\ast)/\gamma $, for $1<\Gamma \le 2$). Thus we have, similarly as before,
\begin{eqnarray}
\dot m & = & \pi m^2_\ast \frac{m - m_\ast}{\gamma}   \left(\frac{a_\ast}{a_\infty} \right)^{\frac{5-3\Gamma}{\Gamma-1}}\left(\frac{\Gamma-1-a^2_\infty}{\Gamma-1-a^2_\ast} \right)^{\frac{1}{\Gamma-1}} \nonumber \\
& & \times \frac{(\Gamma - 1)(1+3 a^2_\ast)}{a_\infty^3 (\Gamma - 1 - a^2_\ast)}.
\label{mdotnnowe2}
\end{eqnarray}
From that we get the  time $T(m_\mathrm{in},M)$ (the proper time of a stationary asymptotic observer)
\begin{eqnarray}
\label{dotTn}
\lefteqn{T(m_\mathrm{in},M) =} \nonumber \\
&& \frac{\gamma}{\pi} \left( \frac{a_\infty^2}{a_\ast^2} 
\right)^\frac{(5 - 3 \Gamma)}{2(\Gamma - 1)}  \left(\frac{\Gamma-1-a^2_\ast}{\Gamma-1-a^2_\infty} \right)^{\frac{1}{\Gamma-1}} \nonumber \\
&&\times \frac{a^3_\infty (\Gamma - 1 - a^2_\ast)}{(\Gamma - 1) (1 + 3 a^2_\ast)} \int_{m_{\mathrm{in}}}^M \frac{dm_\ast}{m_\ast^2  \left(m-m_\ast\right)},
\end{eqnarray}
the interval of time needed for the increase of mass within a sphere of the areal radius $R_\ast$ from initial $m_\ast$ to final $M-m_\mathrm{in}+m_\ast$. The mass of the central black hole would increase from initial $m_\mathrm{in}$ to $M $ during the time $T(m_\mathrm{in},M)$. The lifetime $T_\mathrm{L}$  of a black hole can be defined as the time $T$ that is needed for the transfer of the whole initial fluid  mass  $m-m_\mathrm{in} $ into the black hole.
The integral (\ref{dotTn}) can be explicitly evaluated, in our idealized case, with the result
\begin{eqnarray}
\label{dotT2n}
\lefteqn{T(m_\mathrm{in},M) =} \nonumber \\
&&  \frac{\gamma}{\pi} \left( \frac{a_\infty^2}{a_\ast^2} 
\right)^\frac{(5 - 3 \Gamma)}{2(\Gamma - 1)}  \left(\frac{\Gamma-1-a^2_\ast}{\Gamma-1-a^2_\infty} \right)^{\frac{1}{\Gamma-1}} \frac{a^3_\infty (\Gamma - 1 - a^2_\ast)}{(\Gamma - 1) (1 + 3 a^2_\ast)}  \nonumber \\
&& \times \frac{1}{m^2} \left\{ \ln \left[\frac{M\left( m-m_\mathrm{in}\right)}{m_\mathrm{in}\left(m-M\right)}\right] + \frac{m}{m_\mathrm{in}}-\frac{m}{M} \right\}.
\end{eqnarray}

We have similar conclusions as in the case formerly discussed for $p$--$\varrho$ polytropic stars. In the limit $M\rightarrow m$ we get   $T(m_\mathrm{in},m)\rightarrow \infty $; the black hole exists for an infinite time within our idealized polytropic $p$--$n$  star. If one defines  a reasonable cutoff---for instance that the mass ratio $M/m$ is smaller than 0.999---then the lifetime is finite, with the dominant given by the term with $m/m_\mathrm{in}$.

\subsection{Lifetime of primordial black holes within neutron stars}

In their work Richards, Baumgarte, and Shapiro \cite{Baumgarte_Shapiro,BS} assume  that backreaction can be neglected  in primordial black holes within neutron stars.  Unfortunately, this is not true, as we shall explain. 
The size of neutron stars is larger than $9m/4$ (this is the Buchdahl limit \cite{Buchdahl}) and it probably does not exceed $50m$, where $m$ is the mass. Neutron stars are compact objects, so compact that they do not satisfy those assumptions of our theoretical results of Sections V and VI that would guarantee the absence of backreation.  Thus it is not really suprising  that there are numerical examples of Section VII that show, for steadily accreting neutron star configurations with black holes in their centers, that   sonic  parameters $a_\ast$, $U_\ast$, and $m_\ast/(2R_\ast)$ are not constant. The backreaction cannot be ignored under these circumstances. The  analysis sketched in subsections A and B of the present Section VIII cannot be applied in order to estimate the lifetime of black holes within neutron stars. Thus, literally speaking, results of Baumgarte and Shapiro on the lifetime of black holes within neutron stars, cannot be regarded as being proven.

On the other hand, the plots of the mass accretion rate $\dot m$ shown in Figs.\ \ref{fig:13} and \ref{fig:15}---where backreaction manifests quite strongly---are quite similar to the plot of $\dot m$ in Fig.\ \ref{fig:17}, in which backreaction is clearly negligible. That might mean that backreaction does not change dramatically the time $T(m_\mathrm{in},M)$. This issue warrants investigation, in our opinion.

\begin{figure}
\centering
\includegraphics[width=1\columnwidth]{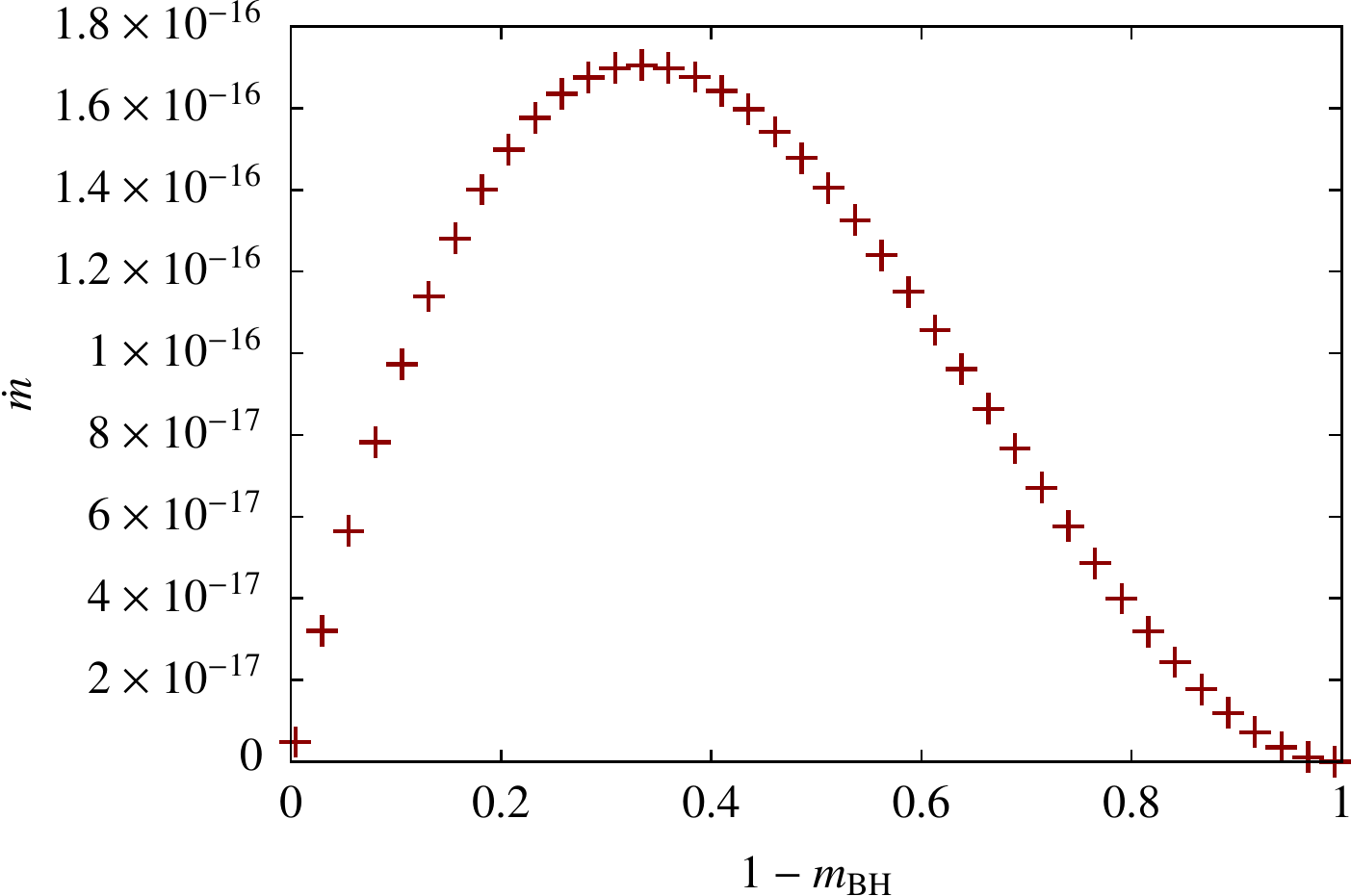}
\caption{The mass accretion rate $\dot m$ (ordinate) versus the fluid mass $1 - m_\mathrm{BH}$ (abscissa). The asymptotic parameters are: $m = 1$, $a^2_\infty = 10^{-2}$, $\Gamma = 1.5$, and $R_\infty = 10^6$.}
\label{fig:17}   
\end{figure} 

\subsection{Other applications}

A review of our key results in Sections V and VI demonstrates, that   it is legitimate to use them in the case of a white dwarf, in which matter is accreting onto a wandering black hole, that found itself in its center. The astrophysical systems of interest can also consist of stars harbouring black holes---a sub-type of Thorne--\.Zytkow stars \cite{Zytkow1, Zytkow2}. In these objects one can perform  the analysis outlined in subsection VIII A and VIII B.

\section{Summary}  

It has been noted almost two decades ago that, under suitable conditions, backreaction does not influence intensive  characteristics of sonic points. This paper deals with wider classes of polytropic equations of state, with polytropic indices in the range $(1,2]$. We derive boundary conditions that allow one to prove analytically that intensive characteristics of critical sonic points are in fact the same as for test fluids. We provide  numerical examples that confirm the validity of analytic proofs. More importantly, we find numerical solutions of accreting systems, in which all characteristics of sonic points of critical flows depend on selfgravity.  

This fact has a consequence---the recently announced analysis of lifetimes of primordial black holes within neutron stars \cite{Baumgarte_Shapiro, BS} assumes the existence of intensive parameters  of the related accretion. We have found numerical counterexamples to this assumption.

\begin{acknowledgments}
We would like to acknowledge and thank Janusz Karkowski for his help in designing the numerical method used in this paper. P.\ M.\ was partially supported by the Polish National Science Centre Grant No.\ 2017/26/A/ST2/00530.
\end{acknowledgments}


\begin{thebibliography}{99}
\bibitem{remark} In the test fluid approximation the parameters of the sonic point do not depend on the central mass. See Theorem 2 and Eqs.\ (5.10--5.11) in \cite{malec}.
\end{thebibliography}
\end{document}